\documentclass[reqno,11pt]{amsart}

\usepackage{graphicx}         
\usepackage{amsmath}
\usepackage{amsfonts}
\usepackage{amssymb}
\usepackage{eucal}
\usepackage[latin1]{inputenc}
\usepackage[all]{xy}
\usepackage{color}

\newtheorem{thm}{Theorem}[section]
\newtheorem{cor}[thm]{Corollary}
\newtheorem{lemma}[thm]{Lemma}

\newtheorem{prop}[thm]{Proposition}
\theoremstyle{definition}
\newtheorem{defn}[thm]{Definition}
\theoremstyle{remark}

\newtheorem{example}[thm]{Example}

\renewcommand{\b}{\beta}

\newcommand{\e}{\epsilon}

\renewcommand{\l}{\lambda}
\newcommand{\s}{\sigma}

\newcommand{\fr}[2]{\left[\frac{#1}{#2}\right]}

\newcommand{\PB}{\left\{\cdot\,,\cdot\right\}}

\newcommand{\Pb}[1]{\left\{\cdot\,,#1\right\}}
\newcommand{\pb}[1]{\left\{#1\right\}}

\renewcommand{\(}{\left(}
\renewcommand{\)}{\right)}
\renewcommand{\[}{\left[}
\renewcommand{\]}{\right]}

\newcommand{\set}[1]{\left\{#1\right\}}

\newcommand{\cD}{\mathcal D}
\newcommand{\D}[1]{\cD_{#1}}
\newcommand{\Dp}[1]{\cD_{(#1)}}

\newcommand{\X}{\mathcal X}

\newcommand{\cR}{\mathcal R}
\newcommand{\cS}{\mathcal S}

\newcommand{\bbC}{\mathbb C}

\newcommand{\bbN}{\mathbb N}

\newcommand{\bbR}{\mathbb R}

\newcommand{\bbZ}{\mathbb Z}

\newcommand{\F}[1]{F^{(#1)}}
\newcommand{\B}[1]{B^{(#1)}}
\newcommand{\nk}{{(n,k)}}
\newcommand{\nkb}{{(n,k),b}}

\newcommand{\um}{{\underline m}}

\newcommand{\LV}[2]{\hbox{LV}(#1,#2)}
\newcommand{\LVB}[2]{\hbox{LV}_b(#1,#2)}

\newcommand{\ds}{\displaystyle}

\newcommand{\leqs}{\leqslant}
\newcommand{\geqs}{\geqslant}
\newcommand{\pp}[2]{\frac{\partial#1}{\partial#2}}
\newcommand{\p}{\partial}

\newcommand{\diff}{{\rm d }}

\newcommand{\Id}{\mathop{\rm Id}}

\renewcommand{\geq}{\geqs}
\renewcommand{\leq}{\leqs}

\newif\ifprivate
\privatefalse

 \numberwithin{equation}{section}

\def\???{\ifprivate {\bf {???}} \marginpar{{\Huge {\bf ?}}}\else \fi}
\numberwithin{equation}{section}

\begin{document}

\nocite{*}

\title[Integrable reductions of the dressing chain]{Integrable reductions of the dressing chain}

\author[C. A. Evripidou]{C. A. Evripidou} \thanks{The research of the first author was supported by the project
  ``International mobilities for research activities of the University of Hradec Kr\'alov\'e'',
  CZ.02.2.69/0.0/0.0/16\_027/0008487} \address{Charalampos Evripidou, Department of Mathematics, Faculty of
  Science, University of Hradec Kralove, Czech Republic} \email{charalambos.evripidou@uhk.cz}

\author[P. Kassotakis]{P. Kassotakis}
\address{Pavlos Kassotakis, Department of Mathematics and Statistics, University of Cyprus, P.O.~Box 20537, 1678 Nicosia, Cyprus}
\email{pavlos1978@gmail.com}

\author[P. Vanhaecke]{P. Vanhaecke} \address{Pol Vanhaecke, Universit\'e de Poitiers, Laboratoire de
  Math\'ematiques et Applications, UMR 7348 du CNRS, B\^at.\ H3, Boulevard Marie et Pierre Curie, Site du Futuroscope,
  TSA 61125, 86073 POITIERS Cedex 9}\email{pol.vanhaecke@math.univ-poitiers.fr}

%

\date{\today}


\begin{abstract}
In this paper we construct a family of integrable reductions of the dressing chain, described in its Lotka-Volterra
form. For each $k,n\in\bbN$ with $n\geqs 2k+1$ we obtain a Lotka-Volterra system $\LVB nk$ on~$\bbR^n$ which is a
deformation of the Lotka-Volterra system $\LV nk$, which is itself an integrable reduction of the $2m+1$-dimensional
Bogoyavlenskij-Itoh system $\LV{2m+1}m$, where $m=n-k-1$. We prove that $\LVB nk$ is both Liouville and
non-commutative integrable, with rational first integrals which are deformations of the rational first integrals of
$\LV{n}{k}$. We also construct a family of discretizations of $\LVB n0$, including its Kahan discretization, and we
show that these discretizations are also Liouville and superintegrable.
\end{abstract}

\maketitle

\setcounter{tocdepth}{2}

\tableofcontents

\section{Introduction}
The dressing chain is an integrable Hamiltonian system, which was constructed in \cite{VS} as a fixed point of
compositions of Darboux transformations of the Schr\"odinger operator. It was shown in \cite{RCD} that after a
simple linear transformation it becomes a Lotka-Volterra system which is a deformation of the Bogoyavlenskij-Itoh
system \cite{itoh1,Bog2}. For the integrable reductions of the dressing chain which we will study here, the
latter formulation is the most convenient; also, we will use many results from \cite{KKQTV,KQV,PPPP,RCD},
which are all written in that formulation.

For integers $n$ and $k$, satisfying $n\geqs 2k+1$, the Hamiltonian system $\LV nk$ has as its phase space
$\bbR^n$, which we equip with its natural coordinates $x_1,\dots,x_n$. It has as Hamiltonian $H$ the sum of these
coordinates, $H:=x_1+\cdots+x_n$, and as Poisson structure a quadratic Poisson structure, with brackets
$$
  \pb{x_i,x_j}:=A_{i,j}^\nk x_ix_j\;,\quad\hbox{where}\quad A_{i,j}^\nk
  =\left\{
  \begin{array}{ll}
    +1\hbox{ if $i+n>j+k\;,$}\\
    -1\hbox{ if $i+n\leqs j+k\;,$}
  \end{array}
  \right.
$$
when $1\leq i<j \leq n$. The Hamiltonian vector field $\X_H$ has the form
\begin{equation*}
  \dot x_i=\sum_{j=1}^n A^{(n,k)}_{i,j}x_ix_j\;,\qquad 1\leqs i\leqs n\;.
\end{equation*}%

These systems, for $k=0$, were introduced in \cite{KKQTV} where their Liouville and superintegrability was
established with rational first integrals. For $n=2k+1$ one recovers the Bogoyavlenskij-Itoh
system \cite{Bog2} whose deformation,
which is the dressing chain \cite{VS}, was constructed in \cite{RCD}. We denote this deformation by $\LVB{2k+1}k$.
The observation that $\LV n0$ can be obtained by a reduction from a Bogoyavlenskij-Itoh system
$\LV{2n-1}{n-1}$ led us to study the more general case $\LV nk$ with $k>0$, since these
systems can be obtained by a similar reduction from a Bogoyavlenskij-Itoh system
$\LV{2m+1}m$, with $m:=n-k-1$. The systems $\LV nk$ for $k>0$, were studied in detail in \cite{PPPP},
where their Liouville and non-commutative integrability was proven (see Definition \ref{def:non-com}),
again with rational first integrals. The same reduction can be applied to the deformed systems $\LVB{2m+1}m$,
leading to Hamiltonian systems, which we will denote by $\LVB nk$. The Hamiltonian vector field
$\X_H$ now has the form
\begin{equation*}
  \dot x_i=\sum_{j=1}^n\(A^{(n,k)}_{i,j}x_ix_j+B^{(n,k)}_{i,j}\)\;,\qquad 1\leqs i\leqs n\;,
\end{equation*}%
where all entries $b_{i,j}$ of the skew-symmetric matrix $B^\nk$, satisfying $\vert i-j\vert\notin\set{m,m+1}$, are
zero and the other entries are arbitrary parameters.  Setting in these systems all deformation parameters equal to
zero, one recovers $\LV nk$. A natural question, studied here, is the integrability of $\LVB nk$ for all $n$ and
$k$ with $n\geq 2k+1$. For $\LVB{2k+1}k$ the answer is known \cite{VS,RCD}: $\LVB {2k+1}k$ is Liouville
integrable with polynomial first integrals which are deformations of the first integrals of $\LV {2k+1}k$.


The main result of this paper is that $\LVB nk$ is on the one hand  Liouville integrable, with rational first integrals
which are deformations of the first integrals of $\LV nk$, and is on the other hand non-commutative integrable, with such
first integrals. See Theorem \ref{thm:rat_int} for the case of $(n,0)$ and Theorem \ref{thm:main} for the case of
$(n,k)$ with $k>0$. In order to establish these results, we need to construct the deformed first integrals and show that
they have the desired involutivity properties; independence is in fact quite automatic and is proven by a simple
deformation argument.

Surprizingly, the construction of the deformed first integrals from the undeformed ones is very simple, and is the
same for all first integrals of $\LV nk$ that were constructed in \cite{PPPP}: from such a first integral $F$ of $\LV nk$
we obtain a first integral $F^b$ of $\LVB nk$ by setting $F^b:=e^{\D b}F=F+\D bF+\frac{\D b^2}2F+\cdots$, where
\begin{equation}
  \D{b}=\sum_{1\leqs i\leqs k+1}b_{i,i+m}\frac{\p^2}{\p x_i\p x_{i+m}}-
         \sum_{1\leqs i\leqs k}b_{i,i+m+1}\frac{\p^2}{\p x_i\p x_{i+m+1}}\;,
\end{equation}%
where we recall that $m=n-k-1$. Notice that $H^b=H$ because
$e^{\D b}$ acts on linear polynomials as identity.

The proof that we get in this way first integrals and that they are in involution when the corresponding
undeformed first integrals are
in involution needs however extra work, as it does not follow directly from their definition.
In the case of $\LVB n0$, studied in Section \ref{sec:rational}, there is only one deformation
parameter $\beta:=b_{1,n}$ and the above action of $e^{\D b}$ on the rational first integrals
of $\LV n0$ which were constructed in \cite{KKQTV} can be equivalently described as the pullback
of a birational map, which we introduce. Moreover, we show that this map is a Poisson map between
the deformed and undeformed systems (Proposition~\ref{prp:birat_Poisson}). This yields the
integrability results for $\LVB n0$, since apart from the Hamiltonian, all constructed first
integrals are rational; the fact that these rational first integrals are in involution with the
Hamiltonian, i.e.\ are first integrals, can in this case be shown by direct computation (Proposition \ref{prp:rat_first_integrals}).

When $k>0$ the above idea can also be used, but some care has to be taken because there are now
$2k+1$ deformation parameters, and they can be added one by one, upon decomposing
$\D b=\sum_{p=1}^{2k+1}\Dp p$, but in order to be able to view at each step the action of
$e^{\Dp p}$ on the rational first integrals as the pullback by some Poisson map, one has to add
the parameters in a very specific order. The reason for this is that in this process the form of
the rational first integrals at each step is very important. With this, one gets that the deformed
rational first integrals of $\LVB nk$ are
in involution (second part of Theorem \ref{thm:defo_rat_invol}). This system has $k+1$ independent
polynomial first integrals, which are by construction in involution, because they are restrictions
to a Poisson submanifold of the involutive first integrals of $\LVB{2m+1}m$, but they also have to
be shown to be in involution with the rational first integrals. This is again done by using the above
Poisson maps, but since these maps do not produce the deformed polynomial first integrals, some extra
arguments which are again very much dependent on the particular structure of the first integrals, are
needed. In the end, this proves Theorem \ref{thm:main} which says that the deformed systems $\LVB nk$
are both Liouville and non-commutatively integrable.

Several, a priori different, discretizations of the dressing chain $\LVB {2m+1}m$ have
been constructed and studied in the literature \cite{Adler1993,
NoumiYamada1998, FordyHone2014, PPC}.  We will construct in the final
section of this paper a class of discretizations of the deformed
Lotka-Volterra systems $\LVB{n}k$, with $k=0$. In order to construct them,
we start from the compatibility conditions of a linear problem associated
with the Lax operator of the dressing chain. Upon reducing these conditions
to $\LVB{n}0$, as in the continuous case, we can easily solve the
compatibility conditions, and hence construct the discrete maps explicitly.
We prove that these discrete maps preserve the Poisson structure, the
Hamiltonian and all rational first integrals of $\LVB{n}0$. These
discretizations are therefore both Liouville and
superintegrable. We show that the Kahan discretization of $\LVB{n}0$ is a particular
instance of the discretizations that we construct, thereby showing that the
Kahan map of $\LVB{n}0$ arises as the compatibility condition of a linear
system. For $k>0$ the reduction can also be performed, leading to
an integrable discretization of $\LVB{n}k$, but the proof is rather long
and complicated, so it will not be given here. It is worthwhile pointing out
that when $k>0$ the Kahan discretization of $\LVB{n}k$ is not a particular
case of this discretization.

The structure of the paper is as follows. We construct in Section \ref{sec:def_systems} the systems $\LVB nk$ as
(Poisson) reductions of the systems $\LVB{2m+1}m$, where $m:=n-k-1$, and we show that the inherited Poisson
structure~$\Pi^\nk_b$ is a deformation of the Poisson structure $\Pi^\nk$ of $\LV nk$. In Section
\ref{sec:rational} we construct rational first integrals of $\LVB n0$ as deformations of the first integrals of
$\LV n0$, which were constructed in \cite{KKQTV}. We show by using a Poisson map, which we also construct, that
half of these first integrals are in involution, establishing both the Liouville and superintegrability of $\LVB
n0$. We also give explicit solutions for this system. In Section \ref{sec:k>0} we treat the more complicated case
of $k>0$, where we prove again Liouville integrability, and also non-commutative integrability. In this case we
use $2k+1$ Poisson maps, which are composed
in a very specific order to obtain the results. In Section \ref{sec:discretizations} we construct a family of
discrete maps for $\LVB{n}{0}$ as compatibility conditions for a linear system, associated to the Lax operator of
$\LVB n0$, and show their Liouville and superintegrability. We show that the Kahan discretization of $\LVB n0$ is a
particular case and deduce from this the Liouville and superintegrability of the Kahan discretization.

\section{The Hamiltonian systems $\LVB nk$}\label{sec:def_systems}
In this section, we construct the polynomial Hamiltonian systems $\LVB nk$. Recall that $n$ and $k$ stand for two
arbitrary integers satisfying $n\geqs 2k+1$. We construct them as reductions of the deformed Bogoyavlenskij-Itoh
systems, which we introduced in \cite{RCD}; in the notation of the present paper, the latter systems are the
systems $\LVB{2m+1}m$, where $m:=n-k-1$.

\subsection{The deformed Bogoyavlenskij-Itoh systems}\label{par:BI_deformed}

We first recall the Bogoyavlenskij-Itoh systems $\LV{2m+1}m$, which have first been introduced by O. Bogoyavlenskij
~\cite{Bog1,Bog2} and Y. Itoh \cite{itoh1}, and their deformations $\LVB{2m+1}m$, which we constructed in
\cite{RCD}. In both cases, the phase space of the system is~$\bbR^{2m+1}$, which is equipped with its natural
coordinates $x_1,\dots,x_{2m+1}$. Since many formulas are invariant under a cyclic permutation of these
coordinates, we view the index of $x$ as being taken modulo $2m+1$, i.e., we set $x_{2m+\ell+1}=x_\ell$ for all
$\ell\in\bbZ$. The Poisson structure $\Pi^m$ of $\LV{2m+1}m$ is constructed from the skew-symmetric Toeplitz
matrix\footnote{Later on, the matrix $A^m$, and similarly the matrix $B^m$ and the Poisson structure $\Pi^m$, will
  have two superscripts; in that notation, $A^{m}$ is written as $A^{(2m+1,m)}$, and similarly for~$B^m$ and
  $\Pi^m$.}  $A^{m}$ whose first row is given by
\begin{equation*}
  (0,\underbrace{1,1,\dots,1}_m,\underbrace{-1,-1,\dots,-1}_m)\;.
\end{equation*}%
It leads to a quadratic Poisson structure $\Pi^m$ on $\bbR^{2m+1}$, upon defining the Poisson brackets
\begin{equation*}
  \pb{x_i,x_j}^m:=A^{m}_{i,j}x_ix_j\;,\qquad 1\leqs i,j\leqs 2m+1\;.
\end{equation*}%
As Hamiltonian we take $H:=x_1+x_2+\cdots+x_{2m+1}$, the sum of all coordinates. The corresponding Hamiltonian
system is given by
\begin{equation}\label{eq:BI}
  \dot x_i=x_i\sum_{j=1}^m(x_{i+j}-x_{i-j})\;,\qquad 1\leqs i\leqs 2m+1\;.
\end{equation}%
It is called the \emph{Bogoyavlenskij-Itoh system}, and is denoted by~$\LV{2m+1}m$. Given any real
skew-symmetric matrix $B^m$ of size $2m+1$, define
\begin{equation}\label{eq:deformed_poisson_BI}
  \pb{x_i,x_j}^m_b:=A^{m}_{i,j}x_ix_j+B^m_{i,j}\;,\qquad 1\leqs i,j\leqs 2m+1\;.
\end{equation}%
These brackets define a Poisson structure, denoted $\Pi^m_b$, if and only if all entries
$b_{i,j}:=B^m_{i,j}$ of $B^m$, with $\left|j-i\right|\notin\set{m,m+1}$ are zero~(see \cite[Prop.\ 3]{RCD}). Under this condition on
$B^m$, we can consider the Hamiltonian system on $\bbR^{2m+1}$ with the same Hamiltonian $H$ and Poisson structure
$\Pi^m_b$. It is given by
\begin{equation}\label{eq:BI_deformed}
  \dot x_i=x_i\sum_{j=1}^m(x_{i+j}-x_{i-j})+b_{i,i+m}-b_{i-m,i}\;,\qquad 1\leqs i\leqs 2m+1\;.
\end{equation}%
It is called the \emph{deformed Bogoyavlenskij-Itoh system}, and is denoted by $\LVB{2m+1}m$. It is
clear that setting all parameters $b_{i,j}$ equal to zero, one recovers $\LV{2m+1}m$. A Lax equation (with spectral
parameter) for (\ref{eq:BI_deformed}) is given by
\begin{equation}\label{eq:bogo_deformed_lax}
  (X+\lambda^{-1}\Delta+\lambda M)^\cdot=[X+\lambda^{-1}\Delta+\lambda M,D-\lambda M^{m+1}]\;,
\end{equation}
where for $1\leqs i,j\leqs 2m+1$ the $(i,j)$-th entry of the matrices $X$ and $M$, and of the diagonal matrices
$\Delta$ and $D$, is given by
\begin{gather*}
  X_{i,j}:=x_i\delta_{i,j+m}\;,\quad \Delta_{i,j}:=b_{i+m,j}\delta_{i,j}\;,\quad M_{i,j}:=\delta_{i+1,j}\;,\\
  D_{i,j}:=-\delta_{i,j}(x_i+x_{i+1}+\cdots+x_{i+m})\;,
\end{gather*}%
and the indices of $b, x$ and $\delta$ are considered modulo $2m+1$.
It generalizes Bogoyavlenskij's Lax equation, which can be recovered from it by putting all $b_{i,j}$ equal to zero,
i.e., by setting $\Delta=0$.

\subsection{The reduced systems}\label{par:BI_defo_red_def}
The systems $\LVB nk$, with $n>2k+1$, are obtained by reduction from $\LVB{2m+1}m$, where $m:=n-k-1>~k$.  Consider
the submanifold $N_n$ of $\bbR^{2m+1}$, defined by $x_{n+1}=x_{n+2}=\cdots= x_{2m+1}=0$. It is a linear space of
dimension $n$ which we identify with $\bbR^n$ and on which we take the restrictions of $x_1,x_2,\dots,x_n$ as
coordinates (without changing the notation).

\begin{prop}\label{prp:poisson_sbmfd}
  The submanifold $N_n$ of $\bbR^{2m+1}$ is a Poisson submanifold of~$\(\bbR^{2m+1},\Pi^m_b\right)$
   if and only if the entries of the skew-symmetric matrix $B^m$ satisfy $b_{i,j}=0$ whenever $n+1\leqs i\leqs 2m+1$.
\end{prop}

\begin{proof}
The submanifold $N_n$ is a Poisson submanifold of~$\(\bbR^{2m+1},\Pi^m_b\right)$  if and only if
all Hamiltonian vector fields $\X_F:=\Pb{F}^m_b$, where $F$ is an arbitrary function on
$\bbR^{2m+1}$, are tangent to $N_n$ at all points of $N_n$. This is equivalent to the vanishing of
$\X_F[x_i]=\pb{x_i,F}^m_b$ at all points of $N_n$, for any $i$ with $n+1\leqs i\leqs 2m+1$.
For such $i$, by the derivation property of the Poisson bracket, we have
\begin{equation*}
  \X_F[x_i]=\pb{x_i,F}^m_b=\sum_{j=1}^{2m+1}\pp F{x_j}\(A_{i,j}^{m}x_ix_j+b_{i,j}\)\;,
\end{equation*}%
which equals $\sum_{j=1}^{2m+1}\pp F{x_j}b_{i,j}$ on $N_n$.
This clearly vanishes, for all functions $F$ on $\bbR^{2m+1}$, if and only if
$b_{i,j}=0$ for all $j$.
\end{proof}

Assuming that $B^m$ verifies the assumptions of Proposition \ref{prp:poisson_sbmfd},
$N_n\simeq\bbR^n$ is a Poisson submanifold of $(\bbR^{2m+1},\Pi^m_b)$, and we can restrict $\Pi_b^m$ (as given by
(\ref{eq:deformed_poisson_BI})) to $\bbR^n$, giving a Poisson structure $\Pi_b^{(n,k)}$, with associated Poisson
bracket
\begin{equation}\label{eq:deformed_reduced_poisson}
  \pb{x_i,x_j}^{(n,k)}_b:=A^{(n,k)}_{i,j}x_ix_j+B^{(n,k)}_{i,j}\;,\qquad 1\leqs i,j\leqs n\;,
\end{equation}%
where $A^{(n,k)}$ and $B^{(n,k)}$ denote the $n\times n$ matrices obtained from $A^{m}$ and $B^{m}$ by removing its
last $2m+1-n$ rows and columns. Said differently, $A^{(n,k)}$ denotes the skew-symmetric $n\times n$ Toeplitz matrix
whose first row is given by
\begin{equation*}
  (0,\underbrace{1,1,\dots,1}_{m=n-k-1},\underbrace{-1,-1,\dots,-1}_k)\;.
\end{equation*}%
Thus, all uppertriangular entries of the skew-symmetric matrix $A^\nk$ are $\pm1$ and $A^\nk_{i,j}=1$ if and only
if $n+i>k+j$. Also, $B^{(n,k)}$ is the skew-symmetric $n\times n$ matrix whose uppertriangular entries
$b_{i,j}:=B^{(n,k)}_{i,j}$ with $j-i\notin\set{m,m+1}$ are zero. So, when $k>0$, the first line of $B^{(n,k)}$ is
given by
\begin{equation*}
  (\underbrace{0,0,\dots,0}_{m=n-k-1},b_{1,m+1},b_{1,m+2},\underbrace{0,0,\dots,0}_{k-1})\;,
\end{equation*}%
while for $k=0$ it has the form $(0,0,\dots,0,b_{1,n}).$ We define $\LVB nk$ to be the Hamiltonian system with
$\Pi^{(n,k)}_b$ as Poisson structure and $H=x_1+x_2+\cdots+x_n$ as Hamiltonian.
Explicitly, the Hamiltonian vector field $\X_H$ of $\LVB nk$ is given by
\begin{equation}\label{eq:BI_reduced_deformed}
  \dot x_i=\sum_{j=1}^n\(A^{(n,k)}_{i,j}x_ix_j+B^{(n,k)}_{i,j}\)\;,\qquad 1\leqs i\leqs n\;.
\end{equation}%
Setting $B^{(n,k)}=0$, one recovers the Hamiltonian system $\LV nk$, in particular its Poisson structure
$\Pi^{(n,k)}$, which was constructed and studied in~\cite{PPPP}. Therefore, the system $\LVB nk$ is a deformation
of the system $\LV nk$.

We show in the following proposition that the above matrices $B^{(n,k)}$ are the only ones for which the brackets
given by \eqref{eq:deformed_reduced_poisson} define a Poisson bracket (on $\bbR^n$).
\begin{prop}\label{prp:poisson_deformation}
  Suppose that $B=(b_{i,j})$ is a skew-symmetric $n\times n$ matrix. Then the brackets, given by
  \begin{equation*}
    \pb{x_i,x_j}_b:=A^{(n,k)}_{i,j}x_ix_j+b_{i,j}\;,\qquad 1\leqs i,j\leqs n\;,
  \end{equation*}%
  define a Poisson structure $\Pi_b$ on $\bbR^n$ if and only if all uppertriangular entries $b_{i,j}$ of $B$, with
  $j-i\notin\set{m,m+1}$ are zero. The rank of $\Pi_b$ is $2\[\frac n2\]$.
\end{prop}
\begin{proof}
Let us denote by $\Pi$ the Poisson structure defined by $A^\nk$, and let us denote the derived Poisson bracket by
$\PB$. The constant Poisson bracket defined by $B$ is denoted by $\PB_B$.
We know already from Proposition \ref{prp:poisson_sbmfd} that if all uppertriangular entries $b_{i,j}$ of $B$, with
$j-i\notin\set{m,m+1}$ are zero, then $\Pi_b$ is the restriction of a Poisson structure to a Poisson submanifold,
hence it is a Poisson structure. We therefore only need to show that if one of these entries $b_{i,j}$ with $i<j$
is non-zero, then $\Pi_b$ is not a Poisson structure. Suppose first that $0<j-i<m$ and $i\neq1$. Then
\begin{gather*}
  \pp{}{x_{i-1}}\[\pb{\pb{x_{i-1},x_i},x_j}_B+\pb{\pb{x_i,x_j},x_{i-1}}_B+\pb{\pb{x_j,x_{i-1}},x_i}_B\]\\
  =A^\nk_{i-1,i}b_{i,j}+A^\nk_{i-1,j}b_{i,j}=2b_{i,j}\neq0\;,
\end{gather*}%
%
so that $\Pi_b$ does not satisfy the Jacobi identity.  When $0<j-i<m$ and $i=1$ it suffices to replace in the above
computation $i-1$ by $j+1$ to arrive at the same conclusion. Finally, when $j-i>m+1$ one replaces in the above
computation $i-1$ by $i+1$ to arrive again at the same conclusion.

The rank of $\Pi$ is the rank of $A^\nk$, which is equal to $n$ when $n$ is even and $n-1$ when $n$ is odd. Since
$\Pi_b$ is obtained by adding constants to the quadratic structure $\Pi$, its rank is at least the rank of
$A^\nk$. However, the rank of $\Pi_b$ is even and bounded by $n$, so $\Pi_b$ and $\Pi$ have the same rank, which
is~$2\[\frac n2\]$.
\end{proof}
The proposition implies that from the above reduction process we get all possible deformations of $\LV nk$ obtained
by adding to $\Pi^\nk$ a constant Poisson structure.

\section{The Liouville and superintegrability of $\LVB n0$}\label{sec:rational}

We construct in this section enough independent first integrals for the Hamiltonian system $\LVB n0$ to prove its
superintegrability and then select from them enough first integrals in involution to prove its Liouville
integrability. Notice that since the phase space of $\LVB n0$ is $\bbR^n$ and since the Poisson structure on it has
rank $n$ or $n-1$, depending on whether $n$ is even or odd, we need to provide $n-1$ independent first integrals to
prove superintegrability and $\[\frac{n+1}2\]$ independent first integrals (including the Hamiltonian) in
involution to prove Liouville integrability. Throughout the section, $n$ is fixed, and $k=0$ also, so we will drop
from  the notations the label  $(n,0)$, except in the statements of the propositions and the theorem.

\subsection{First integrals}\label{par:first_integrals_LVn0}
We first write down the equations for the vector field~$\X_H$ where we recall that $H=x_1+x_2+\cdots+x_n$ and that
the Poisson structure $\Pi_b=\Pi_b^{(n,0)}$ is defined by (\ref{eq:deformed_reduced_poisson}); the matrix
$B=B^{(n,0)}$ has all entries equal to zero, except for $b_{1,n}=-b_{n,1}$, which we will denote in this section
by~$\b$. Also, the skew-symmetric matrix $A=A^{(n,0)}$ has all its uppertriangular entries equal to (plus!)
1. Therefore, $\X_H$ is given by
\begin{eqnarray}
  \dot x_1&=&x_1(x_2+x_3+\cdots+x_n)+\beta\;,\nonumber\\
  \dot x_i&=&x_i(-x_1-\dots -x_{i-1}+x_{i+1}+\cdots+x_n)\;,\quad 1<i<n\;,\nonumber\\
  \dot x_n&=&x_n(-x_1-x_2-\cdots-x_{n-1})-\beta\;.\label{eqn:vf_k=0}
\end{eqnarray}
We construct the first integrals of this system as deformations of the first integrals of $\LV n0$, which were
constructed in \cite[Prop.\ 3.1]{KKQTV}. We first recall the formulas for these first integrals. For $1\leqs
\ell\leqs\left[\frac{n+1}2\right]$, the following functions $F_\ell=F_\ell^{(n,0)}$ are first integrals of $\LV
n0$:
\begin{equation}\label{eq:integrals}
  F_\ell:= \left\{
  \begin{array}{ll}
    \left(x_1+x_2+\cdots+x_{2\ell-1}\right)\ds\frac{x_{2\ell+1}x_{2\ell+3}\ldots x_{n}}{x_{2\ell}x_{2\ell+2}\ldots
      x_{n-1}}\;,& \mbox{ if}\ n\mbox{ is
      odd},\\ \\
    \left(x_1+x_2+\cdots+x_{2\ell}\right)\ds\frac{x_{2\ell+2}x_{2\ell+4}\ldots x_n}{x_{2\ell+1}
      x_{2\ell+3}\ldots x_{n-1}}\;,& \mbox{ if} \ n\mbox{ is even}.
  \end{array} \right.
\end{equation}
More first integrals were constructed by using the anti-Poisson involution $\imath$ on $\bbR^n$, defined by
\begin{equation}
\label{eq:anti_poisson}
  \imath(a_1,a_2,\dots,a_n):=(a_n,a_{n-1},\dots,a_1)\;,
\end{equation}%
which leaves $H$ is invariant, $\imath^*H:=H\circ\imath=H$, so that the rational functions $G_\ell:=\imath^*F_\ell$
($\ell=1,\dots,\left[\frac{n+1}2\right]$) are also first integrals of $\LV n0$. This yields exactly $n-1$ different
first integrals, because when $n$ is even, all $F_\ell$ and $G_\ell$ are different, except for $F_{n/2}=H=G_{n/2}$,
and when $n$ is odd, all $F_\ell$ and $G_\ell$ are different, except for $F_{(n+1)/2}=H=G_{(n+1)/2}$ and
$F_1=G_1$. We recall also that the functions $F_\ell$ are pairwise in involution, just like the functions~$G_\ell$,
and that all these functions are independent, which accounts for the Liouville and superintegrability of $\LV n0$.

In order to construct from these first integrals of $\LV n0$ first integrals of $\LVB n0$ we use the constant
coefficient differential operator $\D{b}$, which we define
by
\begin{equation*}
  \D{b}=\beta\frac{\p^2}{\p x_1\p x_n}\;.
\end{equation*}%
Notice that $\imath^*$ and $\D{b}$ commute, $\imath^*\circ \D{b}=\D b\circ \imath^*$.  As said, we use the
operator~$\D b$ to define some first integrals of $\LVB n0$: we define for $1\leqs
\ell\leqs\left[\frac{n+1}2\right]$ the functions $F_\ell^b=F^{(n,0),b}_\ell$ and $G_\ell^b=G^{(n,0),b}_\ell$ by
\begin{equation}\label{equ:Fn9b}
  F^b_\ell:=e^{\D{b}}F_\ell=F_\ell+\beta\frac{\p^2 F_\ell}{\p x_1\p x_n}\;,\qquad
  G^b_\ell:=\imath^*F^b_\ell=e^{\D{b}}G_\ell\;.
\end{equation}%
We have used that when the operator $\D{b}$ is applied twice to $F_\ell$, the result is zero. This follows from the
fact that the variables $x_1$ and $x_n$ appear linearly in $F_\ell$ (and hence are absent in $\D{b}F_\ell$), as is
clear from (\ref{eq:integrals}).  Explicit formulas for the rational functions $F_\ell^b$, with $1\leqs\ell\leqs
\left[\frac{n-1}2\right]$ are given by
\begin{equation}\label{E:integrals}
  F^b_\ell= \left\{
  \begin{array}{ll}
    \left(x_1+x_2+\cdots+x_{2\ell-1}+\frac{\beta}{x_n}\right)
    \ds\frac{x_{2\ell+1}x_{2\ell+3}\ldots x_{n}}{x_{2\ell}x_{2\ell+2}\ldots x_{n-1}}\;,&
    \mbox{ if}\  n\mbox{ is odd},\\
    \\
    \left(x_1+x_2+\cdots+x_{2\ell}+\frac{\beta}{x_n}\right)
    \ds\frac{x_{2\ell+2}x_{2\ell+4}\ldots x_n}{x_{2\ell+1} x_{2\ell+3}\ldots x_{n-1}}\;,&
   \mbox{ if} \  n\mbox{ is even},
  \end{array} \right.
\end{equation}
and similarly for $G_\ell^b$. Also, for $\ell=\[\frac{n+1}2\]$ the above definitions (\ref{equ:Fn9b}) amount to
$F^b_\ell=G^b_\ell=H$.

\goodbreak
\begin{prop}\label{prp:rat_first_integrals}
  For $\ell=1,\dots,\left[\frac{n+1}2\right]$, the rational functions $ F^{(n,0),b}_\ell$ and $G^{(n,0),b}_\ell$ are
  first integrals of $\LVB n0$.
\end{prop}
\begin{proof}
Since $\imath$ is an anti-Poisson map which leaves the Hamiltonian $H$ invariant, it suffices to show that the
rational functions $ F^b_\ell$ are first integrals of $\LVB n0$. We do this for odd $n$, the case of even $n$ being
completely analogous. Let $1\leqs \ell\leqs\left[\frac{n-1}2\right]$. To prove that $F^b_\ell$ is a first integral
of (\ref{eqn:vf_k=0}) we show that its logarithmic derivative $(\log(F^b_\ell))^\cdot=\dot F_\ell^b/F_\ell^b$ is
zero. Thanks to the particular form of the vector field~(\ref{eqn:vf_k=0}), one easily obtains the following two
formulas:
\begin{eqnarray*}
  \(\log\(x_1+x_2+\cdots+x_{2\ell-1}+\frac{\beta}{x_n}\)\)^\cdot&=&x_{2\ell}+x_{2\ell+1}+\cdots+x_n+
  \frac{\beta}{x_n}\;,\\
  \(\log\(\frac{x_{2\ell+1}}{x_{2\ell}}\cdots\frac{x_n}{x_{n-1}}\)\)^\cdot&=&-x_{2\ell}-x_{2\ell+1}-
  \cdots-x_n-\frac{\beta}{x_n}\;.
\end{eqnarray*}
Summing them up, we find $(\log(F^b_\ell))^\cdot=0$, and hence that $\dot F^b_\ell=0$.
\end{proof}

\subsection{Involutivity}
We now show that the first integrals $F_\ell^b$ of $\LVB n0$ are in involution. For doing this, observe by
comparing (\ref{eq:integrals}) and (\ref{E:integrals}) that formally $F_\ell^b$ can be obtained from $F_\ell$ by
replacing $x_1$ with $x_1+\beta/x_n$. Said differently, if we denote by $\s:\bbR^n\to\bbR^n$ the birational map
defined for $(a_1,\dots,a_n)$ with $a_n\neq0$ by
\begin{equation}\label{eq:birat}
  \s(a_1,\dots,a_n):=\(a_1+\frac{\beta}{a_n},a_2,\dots,a_n\)\;,
\end{equation}
then $\s^*F_\ell=F_\ell^b$.
\begin{prop}\label{prp:birat_Poisson}
  The birational map $\s:\(\bbR^n,\PB_b^{(n,0)}\)\to\(\bbR^n,\PB^{(n,0)}\)$, defined by (\ref{eq:birat}), is a Poisson
  map.
\end{prop}
\begin{proof}
It suffices to show that $\pb{\s^* x_i,\s^*x_j}_b=\s^*\pb{x_i,x_j}$ for $1\leqs i<j\leqs n$. Since $\s^*x_i=x_i$
for $i>1$, this is obvious when $1<i<j$. We therefore only need to verify the formula for $i=1$ and $j>1$. If $1<j<n$ then
$$
  \pb{\s^*x_1,\s^*x_j}_b=\pb{x_1+\frac{\beta}{x_n},x_j}_b=\(x_1+\frac{\beta}{x_n}\)x_j
  =\s^*(x_1x_j)
  =\s^*\pb{x_1,x_j}\;,
$$
where we have used, in the second equality, that $\pb{x_n,x_j}_b=-\pb{x_j,x_n}=-x_jx_n$,
with a minus sign because $j<n$. If $j=n$ then
\begin{gather*}
  \pb{\s^*x_1,\s^*x_n}_b=\pb{x_1+\frac{\beta}{x_n},x_n}_b=\pb{x_1,x_n}_b\\
  =x_1x_n+{\beta}=\s^*(x_1x_n) =\s^*\pb{x_1,x_n}\;.
\end{gather*}
\end{proof}
\begin{cor}\label{cor:rat_invol}
  The rational functions $F_\ell^{(n,0),b}$ defined in (\ref{E:integrals}) are in involution with respect to the
  Poisson bracket $\PB^{(n,0)}_b$. Similarly, the rational functions $G_\ell^{(n,0),b}$ are in involution.
\end{cor}
\begin{proof}
Let $1\leqs \ell,\ell'\leqs\left[\frac{n-1}2\right]$. Then, according to Proposition \ref{prp:birat_Poisson},
\begin{equation*}
  \pb{F_\ell^b,F_{\ell'}^b}_b=\pb{\s^*F_\ell,\s^*F_{\ell'}}_b= \s^*\pb{F_\ell,F_{\ell'}}=0\;,
\end{equation*}%
where we have used in the last step that the functions $F_\ell$ of $\LV n0$ are in involution
\cite[Prop.\ 3.2]{KKQTV}. The fact that the functions $G_\ell^{b}$ are also in involution follows from the
fact that $\imath$ is an anti-Poisson map of $\(\bbR^n,\Pi_b\)$.
\end{proof}
Notice that although $\sigma$ is a birational Poisson isomorphism, it is \emph{not} an isomorphism between the
Hamiltonian systems $\LV n0$ and $\LVB n0$ because $\s^*H\neq H$. In particular, (\ref{cor:rat_invol}) does not
imply that the rational functions $F_\ell^b$ are in involution with the Hamiltonian $H$, i.e., that they are first
integrals of $\LVB n0$; this requires a separate proof, which has been given in Proposition
\ref{prp:rat_first_integrals} above.

\subsection{Integrability}\label{par:rat_integrability}
We now prove the Liouville and superintegrability of $\LVB n0$. As we will see, the main result that remains to be
proven is that the $n-1$ constructed first integrals, to wit the Hamiltonian, the rational functions $F^b_\ell$
(with $\ell=1,\dots,\left[\frac{n-1}2\right]$) and the rational functions $G^b_\ell$ (with
$\ell=1,\dots,\left[\frac{n-1}2\right]$ when $n$ is even and $\ell=2,\dots,\left[\frac{n-1}2\right]$ when $n$ is
odd) are independent, i.e. have independent differentials on an open dense subset of $\bbR^n$. Since these
functions are rational, it suffices to show that their differentials are independent in at least one point of
$\bbR^n$.

To see this, we use the fact that the undeformed functions $H,\ F_\ell$ and $G_\ell$ are independent at some point
$P$ (see \cite{KKQTV}). Since the deformed functions depend polynomially on the deformation parameter $\beta$,
they will still be
independent at $P$ for $\beta$ in a small interval, centered at zero. Notice that if we rescale all variables by a
factor $\lambda\neq0$ and rescale $\beta$ by a factor $\lambda^2$ all these functions also get multiplied by a
non-zero factor. It follows that the differentials of the deformed functions are independent at $P$ for all values
of $\beta$.

\begin{thm}\label{thm:rat_int}
  For any $n$, the Hamiltonian system $\LVB n0$ is superintegrable, with first integrals the rational functions
  $F^{(n,0),b}_\ell$ and $G^{(n,0),b}_\ell$. Moreover, it is Liouville integrable with rational functions
  $F^{(n,0),b}_\ell$, where $\ell=1,\dots,\left[\frac{n+1}2\right]$; also, it is Liouville integrable with rational
  functions $G^{(n,0),b}_\ell$, where $\ell=1,\dots,\left[\frac{n+1}2\right]$.
\end{thm}
\begin{proof}
Recall that a superintegrable system on an $n$-dimensional manifold is a vector field (Hamiltonian or not), which
has $n-1$ independent first integrals. As we have constructed precisely this number of independent first integrals
for $\LVB n0$, we have proven its superintegrability. For Liouville integrability of a Hamiltonian vector field on
an $n$-dimensional Poisson manifold of rank $2r$ we need $n-r$ independent first integrals which are in
involution. Here, the rank of the Poisson structure $\Pi_b$ is $2\[\frac n2\]$ (see
Proposition~\ref{prp:poisson_deformation}) so that we need $n-\[\frac n2\]=\[\frac{n+1}2\]$ such first integrals, which
is exactly the number of independent first integrals $F^b_\ell$ (or $G^b_\ell$) that we have, and they are in
involution by Corollary \ref{cor:rat_invol}. Notice that each of these sets of first integrals contains the
Hamiltonian $H$. Notice also that when $n$ is odd, $F_1^b$ is a Casimir function of $\Pi_b$.
\end{proof}

\begin{example}
For $n=4$ and $k=0$ the matrices $A^{(4,0)}$ and $B^{(4,0)}$ are given by
\begin{equation*}
  A^{(4,0)}=
  \begin{pmatrix}
    0&1&1&1\\
    -1&0&1&1\\
    -1&-1&0&1\\
    -1&-1&-1&0
  \end{pmatrix}\;,\quad
  B^{(4,0)}=
  \begin{pmatrix}
    0&0&0&\b\\
    0&0&0&0\\
    0&0&0&0\\
    -\b&0&0&0
  \end{pmatrix}\;.
\end{equation*}%
The corresponding system $\LVB40$ is given by the formulas
\begin{eqnarray*}
  \dot x_1&=&x_1(x_2+x_3+x_4)+\b\;,\\
  \dot x_2&=&x_2(-x_1+x_3+x_4)\;,\\
  \dot x_3&=&x_3(-x_1-x_2+x_4)\;,\\
  \dot x_4&=&x_4(-x_1-x_2-x_3)-\b\;,
\end{eqnarray*}
and besides the Hamiltonian $H=x_1+x_2+x_3+x_4$ it has two more independent rational
first integrals $F^b$ and $G^b$, namely
$$
F^b=\frac{(x_1+x_2)x_4+\b}{x_3} \quad \text{and} \quad
G^b=\imath^*F^b=\frac{(x_4+x_3)x_1+\b}{x_2}\;,
$$ where $\imath$ is the anti-Poisson map defined in \eqref{eq:anti_poisson}.  The above three functions give the
superintegrability of the system $\LVB 40$. The rank of the Poisson structure is $4$ and each one of the pairs
$(H,F^b)$ and $(H,G^b)$ provide the Liouville integrability of $\LVB 40$.
\end{example}

\subsection{Explicit solutions}\label{par:solutions}
The Hamiltonian vector field $\X_H$ of $\LV n0$ can be explicitly integrated in terms of elementary functions, as
was first shown in \cite{KQV}. We show that such an integration can also be done for (\ref{eqn:vf_k=0}), the
Hamiltonian vector $\X_H$ of $\LVB n0$. This is most easily done by introducing some
linear coordinates on $\bbR^n$: for $i=0, 1,\dots,n$, let $u_i:=x_1+x_2+\cdots+x_i$ and notice that $u_0=0$ and
$u_n=H$. It is clear that everything can be
easily expressed in terms of the coordinates $u_i$ by substituting $u_{i}-u_{i-1}$ for $x_i$ ($i=1,\dots,n$). As we
will see, this simplifies some of the formulas (for $\X_H$, for example) and makes others more complex (the rational
integrals, for example).
For the proposition which follows, the formulas are the simplest when expressed
in the $u_i$ coordinates.

First, we need to  express $\LVB n0$ in terms of the coordinates $u_1,\dots,u_n$. The simplest way to do this is to
first compute $\Pi_b$ in terms of these coordinates. Since for $i<j$, $\pb{x_i,x_j}_b=x_ix_j$, except that
$\pb{x_1,x_n}=x_1x_n+\b$, we get
\begin{eqnarray}\label{eq:pb_in_u}
  \pb{u_i,u_j}_b&=&u_i(u_j-u_i)\;,\qquad \qquad\hbox{if $1\leqs i<j<n$}\;,\nonumber\\
  \pb{u_i,u_n}_b&=&u_i(u_n-u_i)+\beta\;, \qquad\hbox{if $1\leqs i<n$}\;.
\end{eqnarray}%
Since $H=u_n$, we can compute $\X_H$ as $\Pb{u_n}_b$, which takes in view of the above formulas the following
simple, decoupled form:
\begin{eqnarray}\label{eq:LV0_u}
  \dot{u}_i&=&u_i(H-u_i)+\beta, \qquad i=1, 2, \ldots, n-1\;,\nonumber\\
  \dot{u_n}&=&0\;,
\end{eqnarray}
which can easily be integrated, for any initial condition.  We describe the integration in a geometrical language,
which will be useful when we use it in Section \ref{sec:discretizations}. For any point $P\in\bbR^n$, we can
consider the integral curve of $\X_H$, starting from $P$, which we will denote by $\gamma_P$. Usually, the domain
of an integral curve is taken to be an interval, but in the present case we will take it to be all of $\bbR$ minus
a discrete subset. On the one hand, it is natural to do this because in the case of $\X_H$ the solutions are
precisely defined on such a set. On the other hand, the systems $\LVB nk$ can equally be defined on a complex phase
space $\bbC^n$ and then the integral curves, with complex time, are defined for all of $\bbC$, minus a discrete
subset; the domain of the real integral curves which we consider is just the real part of this complex
subset. Since it is convenient to express the integral curves in terms of coordinates (here the $u_i$ coordinates)
we will write, once $P$ has been fixed, $u_i(t)$ for $u_i(\gamma_P(t))$. The $u$-coordinates of $P$ will be denoted
$(P_1,\dots,P_n)$, so $P_i=u_i(P)$ for $i=1,\dots,n-1$ and $P_n=u_n(P)=H(P)$.
\begin{prop}\label{prp:solutions}
  Let $P$ be any point of $\bbR^n$ and let $\gamma_P$ denote the integral curve of (\ref{eq:LV0_u}), which is $\LVB
  n0$, expressed in the $u_i$ variables, starting from $P$.  Denote by $h$ the value of the Hamiltonian at $P$,
  i.e., $h=H(P)=u_n(P)$. Let $\Delta_0$ be a square root of ${h^2+4\beta}$, which may be real or imaginary. Then,
  for $i=1,2, \ldots, n-1$,
  \begin{equation}\label{eq:u_solution}
    u_i(t)=\left\{
      \begin{array}{ll}
        P_i\;, & \text{if }  P_i^2-P_ih-\beta=0\;,\\ \\
        \frac{h}{2}+\frac{2P_i-h}{2+t(2P_i-h)}, \quad &\text{if } \Delta_0^2=h^2+4\b=0\;,\\ \\
        \frac{(h+\Delta_0)(h-\Delta_0-2P_i)-(h+\Delta_0-2P_i)(h-\Delta_0)e^{-t\Delta_0}}
        {2(h-\Delta_0-2P_i)-2(h+\Delta_0-2P_i)e^{-t\Delta_0}}\;, \ &\text{otherwise}.
      \end{array}
      \right.
  \end{equation}
  Obviously, $u_n(t)=H$ is constant.
\end{prop}



\section{The Liouville and non-commutative integrability of $\LVB nk$}\label{sec:k>0}
In this section, we generalize the results of Section \ref{sec:rational} on the integrability of $\LVB n0$ to the
case of $\LVB nk$, where $k\in\bbN$ satisfies $2k+1<n$, but is otherwise arbitrary. We do not treat here the case
of $n=2k+1$ because we have already established the Liouville integrability of $\LVB{2k+1}k$ in \cite{RCD}. We show
in this section that if $1<2k+1<n$ then $\LVB nk$ is on the one hand Liouville integrable, and on the other hand is
non-commutatively integrable of rank $k+1$. We start by recalling the definition of non-commutative integrability
(see \cite{Mich_Fom,PLV}), which we specialize to $\bbR^n$.
\begin{defn}\label{def:non-com}
Let $\Pi$ be a Poisson structure on $\bbR^n$, with associated Poisson bracket $\PB$. Let $\mathbf
F=(f_1,\dots,f_s)$ be an $s$-tuple of functions on $\bbR^n$, where $2s\geqs n$ and set $r:=n-s$. Suppose the
following:
\begin{enumerate}
  \item[(1)] The functions $f_1,\dots,f_r$ are in involution with the functions $f_1,\dots,f_s$:
  $$ \{f_i,f_j\}=0,\qquad 1\leqs i\leqs r \hbox{ and } 1\leqs j \leqs s\;;$$
  \item[(2)] For $P$ in a dense open subset of $\bbR^n$:
  $$ \diff f_1(P)\wedge \dots \wedge \diff f_s(P)\ne 0 \quad \hbox{and} \quad \X_{f_1}|_P\wedge \dots \wedge
    \X_{f_r}|_P\ne 0\;.  $$
\end{enumerate}
Then the triplet $(\bbR^n,\Pi,\mathbf F)$ is called a \emph{non-commutative integrable system} of \emph{rank} $r$.
\end{defn}
The classical case of a \emph{Liouville integrable system} corresponds to the particular case where $r$ is half the
(maximal) rank of $\Pi$; this implies that \emph{all} the functions $f_1,\dots,f_s$ are pairwise in involution. The
case of a superintegrable system corresponds to $r=1$; in this case, setting $H=f_1$, condition (1) just means
that $\X_H$ has $n-1$
first integrals, while the second condition in (2) is trivially satisfied: superintegrability means, as recalled in
the previous section, that $\X_H$ has $n-1$ independent first integrals.

In order to establish Liouville and non-commutative integrability in Section \ref{par:gen_defo_integ} below, we
first construct a set of polynomial first integrals for $\LVB nk$, which are pairwise in involution, and then we
construct a set of rational first integrals for $\LVB nk$, which are also pairwise in involution. This will be done
in the two subsections which follow. Throughout the section, we suppose that $1< 2k+1< n$ and that
$B=B^{(n,k)}$ is a skew-symmetric $n\times n$ matrix such that (\ref{eq:deformed_reduced_poisson}) defines a
Poisson structure on~$\bbR^n$. Recall that this means that the uppertriangular entries $b_{i,j}$ of $B$ with
$j-i\notin\set{m,m+1}$ are zero.

\subsection{The polynomial first integrals}\label{par:pol_int}
Recall from Section \ref{par:BI_defo_red_def} that the systems $\LVB nk$ are obtained by reduction from the systems
$\LVB{2m+1}m$, where $m:=n-k-1>~k$, where the last inequality comes from our assumption $n>2k+1$. Recall also that
in order to do this reduction, one supposes that the last $2m+1-n$ rows and columns of the $(2m+1)\times(2m+1)$
matrix $B$ are zero, so that $B$ can be viewed as an $n\times n$ matrix by removing these zero rows and columns
(see Proposition \ref{prp:poisson_sbmfd}).  Since $\LVB{2m+1}m$ is Liouville integrable, with $m+1$ independent
polynomial first integrals, whose formulas are recalled below, we obtain by reduction a set of first integrals of $\LVB
nk$, which are automatically in involution with respect to the reduced Poisson structure, which is by definition
the Poisson structure $\Pi^{(n,k)}_b$ of $\LVB nk$. One has however to be careful with the independence of the
reduced first integrals, for example some of these reduced first integrals are zero!  Moreover, since $n>2k+1$,
more first integrals are needed for integrability, as we will see.

Let us first recall the formulas for the (polynomial) first integrals of $\LVB{2m+1}m$. One method
of constructing them is as coefficients of the characteristic polynomial of the Lax operator
$L(\l):=X+\lambda^{-1}\Delta+\lambda M$, which we recalled in (\ref{eq:bogo_deformed_lax}). It is a
classical fact that the coefficients of the characteristic polynomial of a Lax operator yield first
integrals for any Lax equation in which the operator appears \cite[Sect.\ 12.2.5]{PLV}. For $L(\l)$,
the following expansion of its characteristic polynomial was obtained in \cite[Prop.\ 8]{RCD}:
\begin{equation}\label{eq:char_poly_Lax}
  \det(L(\l)-\mu\Id)=\l^{2m+1}+\frac{1}{\lambda^{2m+1}}\prod_{j=1}^{2m+1}(b_{j+m,j}-\l\mu)+
    \sum_{i=0}^m (\l\mu)^{m-i}K_i^b\;.
\end{equation}%
Thus, the polynomials $K_i^b$ in this expansion are first integrals of $\LVB{2m+1}m$. Setting the deformation
parameters equal to zero, one recovers the first integrals, $K_i$, of $\LV{2m+1}m$ which were first constructed by
Bogoyavlenskij~\cite{Bog2} and Itoh \cite{itoh2}.
We construct $k+1$ first integrals of $\LVB nk$ by setting, for $i=0,\dots,k,$
\begin{equation}\label{eq:pol_int}
  K^{(n,k),b}_i:={K_i^b}_{\big\vert\bbR^n}:={
  K_i^b}_{\big\vert x_{n+1}=x_{n+2}=\cdots=x_{2m+1}=0}\;,
\end{equation}%
where the notation introduced by the latter equality is a convenient shorthand. By construction, these polynomials
are first integrals of $\LVB nk$ and they are in involution. Also, $K_0^b=H$ and $K_i^b$ is of degree $2i+1$ for
$i=0,1,\dots,k$.

We give an alternative description of the first integrals (\ref{eq:pol_int}) as deformations of the polynomial
first integrals of $\LV nk$. On the one hand, this description will be important for showing the independence of
these first integrals, and on the other hand it will provide information about the structure of these first integrals,
which we will use to prove some of their properties (involutivity, for example).

It was shown in \cite[Prop.\ 9]{RCD} that the first integrals $K_i^b$ of $\LVB{2m+1}m$ can be obtained using the
operator $\D b$, defined by
\begin{equation}\label{eq:DB_operator}
  \D{b}:=\sum_{1\leqs i\leqs 2m+1}b_{i,i+m}\frac{\p^2}{\p x_i\p x_{i+m}}\;,
\end{equation}%
by the following formula, valid for $i=0,1,\dots,m$:
$$
  K_i^b=e^{\D b}K_i=K_i+\D bK_i+\frac{1}{2!}\D b^2K_i+\cdots+\frac{1}{i!}\D b^i K_i\;.
$$
We have used in the last step that $\deg K_i=2i+1$ (see Formula \eqref{eq:k_i_itoh_rest} below).
Let us show that $\D b $ commutes with restriction to $\bbR^n$. Let $F$ be a smooth or rational function on
$\bbR^{2m+1}$. In view of the conditions on $B=B^{(n,k)}$ (see Proposition~\ref{prp:poisson_deformation}),
the operator $\D b$ is given by
\begin{equation}\label{eq:DB_operator_bis}
  \D{b}:=\sum_{1\leqs i\leqs k+1}b_{i,i+m}\frac{\p^2}{\p x_i\p x_{i+m}}-
         \sum_{1\leqs i\leqs k}b_{i,i+m+1}\frac{\p^2}{\p x_i\p x_{i+m+1}}\;,
\end{equation}%
and we see that $\D b$ does not involve derivation with respect to any of the variables
$x_{n+1},\ x_{n+2},\dots,\ x_{2m+1}$ (recall that $n=m+k+1$), so $\D b$ commutes with restriction to the subspace
$\bbR^n$ of $\bbR^{2m+1}$, which is defined by $x_{n+1}=x_{n+2}=\cdots=x_{2m+1}=0$. It follows that
\begin{eqnarray}
  K^{(n,k),b}_i&=&{K_i^b}_{\big\vert\bbR^n}=\(e^{\D b}K_i\)_{\big\vert\bbR^n}
  =e^{\D b}{K_i}_{\big\vert\bbR^n}=e^{\D b}K^{(n,k)}_i\qquad\label{eq:K_are_defos}\\
  &=&K^{(n,k)}_i+\D bK^{(n,k)}_i+\frac{1}{2!}\D b^2K^{(n,k)}_i+\cdots+
  \frac{1}{i!}\D b^i K^{(n,k)}_i\;.\nonumber
\end{eqnarray}
This shows that the polynomial first integrals $K_i^{(n,k),b}$ of $\LVB nk$ are deformations of the first integrals
$K_i^{(n,k)}$ of $\LV nk$. Notice also that (\ref{eq:K_are_defos}) implies that $K_i^{(n,k),b}=0$ when $i>k$ since
$K_i^{(n,k)}=0$ when $i>k$ (see the comments after Proposition 3.3 in \cite{PPPP}). That is the reason why we
restricted $i$ in (\ref{eq:pol_int}) to $i=0,1,\dots,k$ rather than $i=0,1,\dots,m$.

For later use, we quickly recall from \cite{PPPP} a combinatorial formula for $K_i^{(n,k)}$. Let
$\um=(m_1,m_2,\dots,m_{2i+1})$ be a $2i+1$-tuple of integers, satisfying $1\leqs m_1<m_2<\cdots<m_{2i+1}\leqs
n$. We view these integers as indices of the rows and columns of $A^{(n,k)}$: we denote by $A^{(n,k)}_{\um}$ the
square submatrix of $A^{(n,k)}$ of size $2i+1$, corresponding to rows and columns $m_1,m_2,\dots,m_{2i+1}$
of~$A^{(n,k)}$, so that
\begin{equation*}
  (A^{(n,k)}_{\um})_{s,t}=(A^{(n,k)})_{m_s,m_t}\;, \hbox{ for } s,t=1,\dots,2i+1\;.
\end{equation*}
Letting
\begin{equation*}\label{eq:S_def}
  \cS^{(n,k)}_i:=\set{\um\mid A^{(n,k)}_{\um}=A^{(2i+1,i)}}\;,
\end{equation*}
the first integral $K_i^{(n,k)}$ is given by
\begin{equation}\label{eq:k_i_itoh_rest}
  K_i^{(n,k)}=\sum_{\um\in \cS^{(n,k)}_{i}} x_{m_1}x_{m_2}\dots x_{m_i}\dots x_{m_{2i+1}}\;.
\end{equation}
One immediate consequence is that every variable $x_j$ has degree at most one in $K_i^{(n,k)}$
and also in $K_i^{(n,k),b}$.

\goodbreak

\subsection{The rational first integrals}\label{par:rat_int}

We now construct a set of rational first integrals of $\LVB nk$. In order to follow some of the more technical arguments
in this subsection, the reader is advised to already take a look at Section \ref{par:examples} below, where explicit
formulas for a few examples are given.

Recall that we assume in this section that $k>0$ and that $n>2k+1$. We define the rational first integrals of $\LVB nk$
as deformations of the rational first integrals of $\LV nk$, which were first constructed in \cite{PPPP}. We first recall
the definition of the latter first integrals as pullbacks of the rational first integrals $F_\ell^{(n-2k,0)}$ and
$G_\ell^{(n-2k,0)}$ of $\LV {n-2k}0$, which we recalled in Section \ref{par:first_integrals_LVn0}. Consider the
polynomial map $\phi_k:\bbR^n\to\bbR^{n-2k}$, defined by
\begin{equation*}
  \phi_k(a_1,a_2,\dots,a_n):=a_1a_2\dots a_k(a_{k+1},a_{k+2},\dots,a_{n-k})a_{n-k+1}\dots a_{n-1}a_n\;.
\end{equation*}
If we denote the standard coordinates on $\bbR^n$ by $x_1,\dots,x_n$, and on $\bbR^{n-2k}$ by $y_1,\dots,y_{n-2k}$,
then $\phi^*y_i=x_1x_2\dots x_kx_{k+i}x_{n-k+1}\dots x_n$ for $i=1,\dots,n-2k$. It was shown in \cite{PPPP} that
for $\ell=1,2,\dots,\fr{n+1}2-k$, the rational functions $F_\ell^{(n,k)}:=\phi^*_kF_\ell^{(n-2k,0)}$ and
$G_\ell^{(n,k)}:=\phi^*_kG_\ell^{(n-2k,0)}$ are first integrals of $\LV nk$ and that the first integrals $F_\ell^{(n,k)}$
are pairwise in involution with respect to $\PB^{(n,k)}$, just like the first integrals $G_\ell^{(n,k)}$. Setting
$s:=2\ell-1$ when $n$ is odd and $s:=2\ell$ when $n$ is even, so that $s$ and $n$ have the same parity, one
computes easily from (\ref{eq:integrals}) that
\begin{equation}\label{eq:rat_integrals_kn}
  F_\ell^{(n,k)}= \prod_{i=1}^k(x_ix_{n+1-i})\sum_{j=1}^sx_{j+k}
  \prod_{t=1}^{\frac{n-s}2-k}\frac{x_{s+k+2t}}{x_{s+k+2t-1}}\;.
\end{equation}
When $\ell=\fr{n+1}2-k$, the last product in this expression reduces to 1 and so $F_\ell^{(n,k)}$ is actually a
polynomial function. Since some of the arguments below which depend on the structure of the first integrals fail for the
polynomial first integrals, we will exclude these first integrals in this section, and so we will throughout this section
only consider the rational first integrals $F_\ell^{(n,k)}$, with $\ell=1,2,\dots,\fr{n-1}2-k$.  Notice also that every
variable or its inverse appears precisely once in (\ref{eq:rat_integrals_kn}), with the $k$ variables
$x_1,\dots,x_{k}$ and the $k+1$ variables $x_{n-k},\dots,x_n$ appearing linearly. This property will be important
in what follows.

We construct the rational first integrals of $\LVB nk$ as deformations of the first integrals $F_\ell^\nk$ by using the
operators $\D b$ (see (\ref{eq:DB_operator}) or (\ref{eq:DB_operator_bis})): for $\ell=1,2,\dots,\fr{n-1}2-k$, we set
\begin{equation*}
  F_\ell^\nkb:=e^{\D b} F_\ell^\nk=e^{\D b}\(\phi_k^*F_\ell^{(n-2k,0)}\)\;,
\end{equation*}%
and similarly for $G_\ell^\nkb$. The present subsection is devoted to the proof of the following theorem, which
says that the deformed first integrals $F_\ell^\nkb$ are first integrals of $\LVB nk$ which are pairwise in involution.
\begin{thm}\label{thm:defo_rat_invol}
  For $1\leqs\ell\leqs\ell'\leqs\fr{n-1}2-k$,
  \begin{equation*}
    \pb{F_\ell^\nkb,H}_b^\nk=0\;\ \hbox{ and }\ \pb{F_\ell^\nkb,F_{\ell'}^\nkb}_b^\nk=0  \;.
  \end{equation*}%
  The same result holds for the rational functions $G_\ell^\nkb$.
\end{thm}
For the proof of Theorem \ref{thm:defo_rat_invol}, we need some extra notation. Since throughout this subsection
$n$ and $k$ are fixed, we will until the rest of the subsection drop $(n,k)$ from the notation, writing $F_\ell^b$
for $F_\ell^\nkb$, writing $B$ for $\B{n,k}$, and so on.  We will need a specific ordering of the entries of the
matrix $B$. Therefore, we label the parameters of $B$ with single indices as follows (recall that $B$ is
skew-symmetric and that its non-zero uppertriangular entries are at positions $(i,j)$ with $j-i=n-k-1$ or
$j-i=n-k$):
\begin{equation*}
  B=\begin{pmatrix}
  \dots&0&b_{2k+1}&-b_{2k}&0&\dots&0\\
  &\dots&0&b_{2k-1}&-b_{2k-2}&\ddots&\vdots\\
  &&&\ddots&\ddots&\ddots&0\\
  &&&&\ddots&{b_3}&-b_{2}\\ \\
  &&&&\dots&0&b_{1}\\ \\
  &&&&&\dots&0\\
  &&&&&\dots&\vdots\\
  \end{pmatrix}
\end{equation*}%
Expressed in terms of a formula,
\begin{equation*}
  b_p:=(-1)^{p+1} b_{k+1-\fr p2,n-\fr{p-1}2}\;.
\end{equation*}%
For $1\leqs p\leqs 2k+1$, we denote by $\B p$ the matrix obtained from $B$ by setting the parameters
$b_{p+1},b_{p+2},\dots,b_{2k+1}$ equal to zero. We also set $\B 0$ equal to the zero matrix. So $\B1$ contains only
the parameter $b_1$ and $\B{2k+1}=B$. The corresponding Poisson structure, which can be obtained from
$\Pi_b=\Pi^\nk_b$ by setting the parameters $b_{p+1},b_{p+2},\dots,b_{2k+1}$ equal to zero, is denoted by $\Pi_{(p)}$. In
particular, $\Pi_b=\Pi_{(2k+1)}$. The associated Poisson bracket is denoted by~$\PB_{(p)}$. We also associate to
$p$ the following constant coefficient differential operator
\begin{equation}\label{eq:Dp}
  \Dp p:=b_p\frac{\p^2}{\p x_{k+1-\fr p2} \p x_{n-\fr{p-1}2}}\;.
\end{equation}%
It is clear that $\D b=\Dp{2k+1}+\Dp{2k}+\cdots+ \Dp{2}+ \Dp{1}$ and, since these operators
commute, that
$$e^{\D b}=e^{\Dp{2k+1}}\circ e^{\Dp{2k}}\circ\cdots\circ e^{\Dp{2}}\circ e^{\Dp{1}}\;.$$
We could of course consider any alternative order of the operators, but we will use the above one for some reason
which will become clear later. Finally, for any smooth or rational function $F$ on $\bbR^n$ we set $\F p:=e^{\Dp
  p}\F{p-1}$ for $1\leqs p\leqs 2k+1$, and $\F0:=F.$ With this notation, $F_\ell^b=F_\ell^{(2k+1)}$.

A crucial fact which we will use is that the partially deformed integral $F_\ell^{(p)}$ is the pullback of
$F_\ell^{(p-1)}$ by the birational Poisson map $\s_p$, defined for $1\leqs i\leqs n$ by
\begin{equation}\label{eq:sigma_p}
  \s_p^*(x_i):=\left\{
  \begin{array}{ll}
    x_i+\frac{b_p}{x_{i+n-k-1}}\delta_{i,k+1-\frac{p-1}2}&\hbox{ if $p$ is odd,}\\ \\
    x_i+\frac{b_p}{x_{i+k-n}}\delta_{i,n+1-\frac{p}2}&\hbox{ if $p$ is even.}\\
  \end{array}
  \right.
\end{equation}
Notice that
\begin{equation}\label{eq:subs}
  \s_p^*\(x_{n-\fr{p-1}2}x_{k+1-\fr p2}\)=x_{n-\fr{p-1}2}x_{k+1-\fr p2}+b_p\;,
\end{equation}
independently of whether $p$ is even or odd.

%
%
%
%
%
%
%
\begin{prop}\label{prp:pullbacks}
  $\F p_\ell=\s^*_p\F{p-1}_\ell$ for $1\leqs p\leqs 2k+1$ and $1\leqs\ell\leqs\fr{n-1}2-k$.
\end{prop}
\begin{proof}
We first show that if $0\leqs p\leqs 2k$ then
\begin{equation}\label{eq:F_ell^p}
  \F p_\ell=x_{n-\fr p2}x_{k-\fr{p-1}2}E_1+E_2\;,
\end{equation}%
where $E_1$ is independent of $x_{n-\fr p2}$ and of $x_{k-\fr{p-1}2}$, while $E_2$ is independent of $x_{n-\fr
  p2}=x_{n-\frac{p-1}2}$ when $p$ is odd and is independent of $x_{k-\fr{p-1}2}=x_{k+1-\frac p2}$ when $p$ is
even. We do this for $p$ odd, i.e., we show that when $p$ is odd,
\begin{equation*}
  \F p_\ell=x_{n-\frac{p-1}2}x_{k-\frac{p-1}2}E_1+E_2\;,
\end{equation*}%
where $E_1$ is independent of $x_{n-\frac{p-1}2}$ and of $x_{k-\frac{p-1}2}$, while $E_2$ is independent of
$x_{n-\frac{p-1}2}$. According to the explicit formula (\ref{eq:rat_integrals_kn}), we can write $F_\ell=\F0_\ell$
as
\begin{equation*}
  F_\ell=x_{n-\frac{p-1}2}x_{k-\frac{p-1}2}F_\ell'+F_\ell''\;,
\end{equation*}%
where $F_\ell'$ and $F_\ell''$ are independent of $x_{n-\frac{p-1}2}$ and of $x_{k-\frac{p-1}2}$. Then,
\begin{equation*}
  \F p_\ell=e^{\Dp p}\F{p-1}_\ell=\(e^{\Dp p}F_\ell\)^{(p-1)}=\(F_\ell+\Dp pF_\ell\)^{(p-1)}\;,
\end{equation*}%
where we have used in the last step that $\Dp p^2F_\ell=0$, which is a consequence of the fact that $F_\ell$
depends linearly on $x_{k+1-\frac{p-1}2}$ and on $x_{n-\frac{p-1}2}$, which are the variables with respect to which
$\Dp p$ differentiates (see (\ref{eq:Dp})). Since, moreover, $F''_\ell$ is independent of $x_{n-\frac{p-1}2}$,
\begin{equation*}
  \F p_\ell=\(x_{n-\frac{p-1}2}x_{k-\frac{p-1}2}F_\ell'+F_\ell''+ b_px_{k-\frac{p-1}2}
  \frac{\p F'_\ell} {\p x_{k+1-\frac{p-1}2}}\)^{(p-1)}\;.
\end{equation*}%
Since $\Dp 1,\dots,\Dp{p-1}$ do not involve the variables $x_{n-\frac{p-1}2}$ and $x_{k-\frac{p-1}2}$,
\begin{eqnarray*}
  \F p_\ell&=&x_{n-\frac{p-1}2}x_{k-\frac{p-1}2}{F_\ell'}^{(p-1)}+\({F_\ell''}+
  b_px_{k-\frac{p-1}2}\frac{\p F'_\ell} {\p x_{k+1-\frac{p-1}2}}\)^{(p-1)}\\
  &=&x_{n-\frac{p-1}2}x_{k-\frac{p-1}2}E_1+E_2\;,
\end{eqnarray*}
where $E_1$ is independent of $x_{n-\frac{p-1}2}$ and of $x_{k-\frac{p-1}2}$, and $E_2$ is independent of
$x_{n-\frac{p-1}2}$. This shows our claim when $p$ is odd. The proof in case $p$ is even is similar. We use the
obtained formula (\ref{eq:F_ell^p}) to show that $\F p_\ell=\s_p^*\F{p-1}_\ell$ for any $p$. According to
(\ref{eq:F_ell^p}), we can write
\begin{equation*}
  \F {p-1}_\ell=x_{n-\fr {p-1}2}x_{k+1-\fr{p}2}E_1+E_2\;,
\end{equation*}%
where $E_1$ is independent of $x_{n-\fr{p-1}2}$ and of $x_{k+1-\fr{p}2}$, while $E_2$ is independent of $x_{n-\fr
{p-1}2}=x_{n+1-\frac{p}2}$ when $p$ is even and is independent of $x_{k+1-\fr{p}2}=x_{k+1-\frac{p-1}2}$ when $p$
is odd.  Therefore, on the one hand,
\begin{eqnarray*}
  \F p_\ell&=&e^{\Dp p}\F{p-1}_\ell=\F{p-1}_\ell+\Dp p\(x_{n-\fr {p-1}2}x_{k+1-\fr{p}2}E_1+E_2\)\\
  &=&\F{p-1}_\ell+b_p\frac{\p^2}{\p x_{k+1-\fr p2} \p x_{n-\fr{p-1}2}}\(x_{n-\fr {p-1}2}x_{k+1-\fr{p}2}E_1+E_2\)\\
  &=&\F{p-1}_\ell+b_pE_1\;,
\end{eqnarray*}
while on the other hand,
\begin{eqnarray*}
  \s_p^*\F{p-1}_\ell&=&\s_p^*\(x_{n-\fr {p-1}2}x_{k+1-\fr{p}2}E_1+E_2\)\\
  &\stackrel{(\star)}{=}&\s_p^*\(x_{n-\fr {p-1}2}x_{k+1-\fr{p}2}\)E_1+E_2\\
  &\stackrel{(\ref{eq:subs})}=&\(x_{n-\fr {p-1}2}x_{k+1-\fr{p}2}+b_p\)E_1+E_2=\F{p-1}_\ell+b_pE_1\;,
\end{eqnarray*}
which shows that $\F p_\ell=\s^*_p\F{p-1}_\ell$. We have used in $(\star)$ that $\s_p^*(E_1)=E_1$ and that
$\s_p^*(E_2)=E_2$, which hold because $E_1$ and $E_2$ are independent of the only variable which is not fixed by
$\s_p^*$.
\end{proof}
We next show that the maps $\s_p$ are Poisson maps with respect to the appropriate Poisson structures on $\bbR^n$.

\begin{prop}\label{prp:sigma_poisson}
  For $1\leqs p\leqs 2k+1$ the birational map
    $$\sigma_p:\(\bbR^n,\PB_{(p)}\)\to \(\bbR^n,\PB_{(p-1)}\)\;,$$
  defined by (\ref{eq:sigma_p}), is a Poisson map.
\end{prop}
\begin{proof}
We give the proof in case $p$ is even. Then (\ref{eq:sigma_p}) simplifies to
\begin{equation*}
  \s_p^*(x_i)=\left\{
  \begin{array}{ll}
    x_i+\frac{b_p}{x_{i+k-n}}&\hbox{if $i=n+1-\frac p2\;,$}\\ \\
    x_i&\hbox{if $i\neq n+1-\frac p2\;.$}\\
  \end{array}
  \right.
\end{equation*}%
We need to show that $\s_p^*\pb{x_i,x_j}_{(p-1)}=\pb{\s_p^*(x_i),\s_p^*(x_j)}_{(p)}$ for all $1\leqs i<j\leqs
n$. Notice that if $(i,j)\neq\(k+1-\frac p2,n+1-\frac {p}2\)$ then $\pb{x_i,x_j}_{(p)}=\pb{x_i,x_j}_{(p-1)}$ and
otherwise $\pb{x_i,x_j}_{(p)}=\pb{x_i,x_j}_{(p-1)}-b_p$ (recall that $p$ is even). It follows that, if $i,j\neq
n+1-p/2$, then
\begin{equation*}
  \s_p^*\pb{x_i,x_j}_{(p-1)}=\pb{x_i,x_j}_{(p-1)}=\pb{x_i,x_j}_{(p)}=\pb{\s_p^*(x_i),\s_p^*(x_j)}_{(p)}\;,
\end{equation*}%
as was to be shown. Suppose now that $1\leqs i=n+1-\frac p2<j$. Notice that, in this case,
$\pb{x_i,x_j}_{(p)}=\pb{x_i,x_j}_{(p-1)}=+x_ix_j$, with a plus sign. Then
\begin{equation*}%
  \s^*_p\pb{x_i,x_j}_{(p-1)}=\s_p^*(x_ix_j)=\(x_{i}+\frac{b_p} {x_{i+k-n}}\)x_j\;,
\end{equation*}%
while
\begin{equation*}
  \pb{\s^*_p(x_i),\s^*_p(x_j)}_{(p)}=\pb{x_i+\frac{b_p} {x_{i+k-n}},x_j}_{(p)}
  =x_{i}x_j+\frac{b_p} {x_{i+k-n}}x_j\;,
\end{equation*}
where we have used that $\pb{x_{i+k-n},x_j}_{(p)}=-{x_{i+k-n}x_j}$, with a minus sign because $i<j$. Finally,
suppose that $1\leqs i<j=n+1-\frac p2$ and notice that $j+k-n=k+1-\frac p2$. Then,
\begin{eqnarray}\label{eq:poisson_case2}
  \pb{\s^*_p(x_i),\s^*_p(x_j)}_{(p)}&=&\pb{x_i,x_j+\frac{b_p}{x_{j+k-n}} }_{(p)}\nonumber\\
  &=&
  \left\{
  \begin{array}{ll}
    \pb{x_i,x_j}_{(p)}\;,&\hbox{if $i=k+1-\frac p2\;,$}\\ \\
    \pb{x_i,x_j}_{(p)}-\frac{b_p}{x^2_{j+k-n}}\pb{x_i,x_{j+k-n}}\;,&\hbox{if $i\neq k+1-\frac p2\;,$}
  \end{array}
  \right.\nonumber\\
  &=&
  \left\{
  \begin{array}{ll}
    -x_ix_j-b_p\;,&\hbox{if $i=k+1-\frac p2\;,$}\\ \\
    -x_ix_j-\frac{b_px_i}{x_{j+k-n}}\;,&\hbox{if $i<k+1-\frac p2\;,$}\\ \\
    x_ix_j+\frac{b_px_i}{x_{j+k-n}}\;,\qquad&\hbox{if $i>k+1-\frac p2\;,$}
  \end{array}
  \right.\nonumber\\
\end{eqnarray}%
while
\begin{equation*}
  \s^*_p\pb{x_i,x_j}_{(p-1)}=\s^*_p(\pm x_ix_j)=\pm x_i\left(x_j+\frac{b_p}{x_{j+k-n}}\right)\;,
\end{equation*}%
where the $+$ sign corresponds to the case $i>k+1-\frac p2$ and the $-$ sign to the case $i\leqs k+1-\frac
p2$. Clearly, this gives the same result as in (\ref{eq:poisson_case2}). This shows that $\sigma_p$ is a Poisson
map.
\end{proof}
Propositions \ref{prp:pullbacks} and \ref{prp:sigma_poisson} imply, in that order, that for any $p=1,\dots,2k+1$,
and for any $1\leqs\ell\leqs\ell'\leqs\fr{n-1}2-k$,
\begin{equation*}
  \pb{F_\ell^{(p)},F_{\ell'}^{(p)}}_{(p)}=  \pb{\s_p^*F_\ell^{(p-1)},\s_p^*F_{\ell'}^{(p-1)}}_{(p)}=
  \s_p^*\pb{F_\ell^{(p-1)},F_{\ell'}^{(p-1)}}_{(p-1)}
\end{equation*}%
and so, since the undeformed rational first integrals $F_\ell=F_\ell^{(0)}$ are pairwise in involution, an easy induction
shows that their deformations $F^b_\ell=\F{2k+1}_\ell$ are in involution as well. This shows the second part of
Theorem \ref{thm:defo_rat_invol}.

\smallskip

We will next show that that the deformed first integrals $F^b_\ell$ are first integrals of $\LVB nk$. To do this, we
first prove the following lemma.
\begin{lemma}\label{lma:F_ders}
  Let $1\leqs\ell\leqs\fr{n-1}2-k$ and denote $s:=2\ell-1$ when $n$ is odd and $s:=2\ell$ when $n$ is even. If
  $i\notin\set{k+1,k+2,\dots,k+s}$ then $\pb{x_i,F_\ell^b}_b=0$.
\end{lemma}
\begin{proof}
We prove by induction on $p$ that $\pb{x_i,F_\ell^{(p)}}_{(p)}=0$, for $p=0,1,\dots,2k+1$.
Since $F_l^b=F_l^{(2k+1)}$ this proves the statement. For $p=0$ this amounts to
showing that $\pb{x_i,F_\ell}_b=0$, which was done in \cite[Prop.\ 3.1]{PPPP}. Let $1\leqs p\leqs 2k+1$ and assume
that the property is true for $p-1$. If $i$ is such that $\s_p^*(x_i)=x_i$, then by Propositions \ref{prp:pullbacks}
and \ref{prp:sigma_poisson},
\begin{equation*}
  \pb{x_i,F_\ell^{(p)}}_{(p)}=\pb{\s^*_p(x_i),\s^*_p F_\ell^{(p-1)}}_{(p)}
  =\s^*_p\pb{x_i, F_\ell^{(p-1)}}_{(p-1)}=0\;,
\end{equation*}%
where we used the induction hypothesis in the last step. Suppose now that  $\s_p^*(x_i)\neq x_i$.
If $p$ is even, this means that $i=n+1-\frac p2$, and so
\begin{eqnarray*}
  \pb{x_i,F_\ell^{(p)}}_{(p)}&=&\pb{\s_p^*\(x_i-\frac{b_p}{x_{i+k-n}}\),\s_p^*F_\ell^{(p-1)}}_{(p)}  \\
                             &=&\s_p^*\pb{x_i-\frac{b_p}{x_{i+k-n}},F_\ell^{(p-1)}}_{(p-1)}=0\;,
\end{eqnarray*}%
where we used again the induction hypothesis (twice): we could do so because $i+k-n\leqs k$ so that $i+k-n\notin
\set{k+1,k+2,\dots,k+s}.$ If $\s_p^*(x_i)\neq x_i$ and $p$ is odd, then $i=k+1-\frac{p-1}2,$ so that
\begin{eqnarray*}
  \pb{x_i,F_\ell^{(p)}}_{(p)}&=&\pb{\s_p^*\(x_i-\frac{b_p}{x_{i+n-k-1}}\),\s_p^*F_\ell^{(p-1)}}_{(p)}  \\
                             &=&\s_p^*\pb{x_i-\frac{b_p}{x_{i+n-k-1}},F_\ell^{(p-1)}}_{(p-1)}=0\;,
\end{eqnarray*}%
as before.
\end{proof}
Using the lemma and the fact that $F_\ell$ is a first integral of $\LV nk$ (see \cite{PPPP}), we show that
$F_\ell^b$ is a first integral of $\LVB nk$. As before, we show by induction that $\pb{F_\ell^{(p)},H}_{(p)}=0$, the
case of $p=0$ already being established. Suppose that $\pb{F_\ell^{(p-1)},H}_{(p-1)}=0$ for some $p\geqs1$. Then

\begin{eqnarray*}
  \pb{F_\ell^{(p)},H}_{(p)}&=&\pb{\s_p^*F_\ell^{(p-1)},\s_p^*\(H-\frac{b_p}{x_t}\)}_{(p)}\\
                           &=&\s_p^*\pb{F_\ell^{(p-1)},H-\frac{b_s}{x_t}}_{(p-1)}=0\;,
\end{eqnarray*}%
since $t=k+1-\frac n2\leqs k$ when $p$ is even and $t=n-\frac{p-1}2\geqs n-k$ when $p$ is odd; in either case,
$t\notin \set{k+1,k+2,\dots,k+s}$, which proves the last equality. We conclude that $\pb{F_\ell^b,H}_b=0$, which is
the first statement of Theorem \ref{thm:defo_rat_invol}.


\subsection{Integrability}\label{par:gen_defo_integ}
We have now most ingredients to state and prove the Liouville and non-commutative integrability of $\LVB nk$, where
we recall that $k>0$ and $2k+1<n$. Since in this subsection $(n,k)$ is fixed, we will again drop $(n,k)$ from the
notation, except in the statement of the propositions and of the theorem. We have constructed in Section
\ref{par:pol_int} a set of polynomial first integrals for $\LVB nk$ and in Section \ref{par:rat_int} a set of rational
first integrals $F_\ell$. We first show that these polynomial first integrals are in involution with these rational
first integrals. To do this, we need a property of the polynomial first integrals, which we first define.

\begin{defn}
  A polynomial function $K$ on $\bbR^n$ is said to be \emph{$(n,k)$-admissible} if
    \begin{enumerate}
    \item[(1)] $K$ is of degree at most one in each of its variables $x_j$;
    \item[(2)] $K$ can be written (uniquely) as $K=LK'+K''$, where $K''$ is independent of $x_{k+1},\dots,x_{n-k}$
      and $L$ is the sum of these variables, $L=x_{k+1}+x_{k+2}+\dots+x_{n-k}$.
    \end{enumerate}
\end{defn}
A key property of the polynomial Hamiltonians is that they are $(n,k)$-admissible:
\begin{prop}\label{prp:K_i_admissible}
  For $i=0,1,\dots,k$, the polynomial first integral $K_i^\nkb$ is $(n,k)$-admissible.
\end{prop}
\begin{proof}
As said, we write in the proof $K_i$ for $K_i^\nk$ and $K_i^b$ for $K_i^\nkb$. The fact that $K_i$ is
$(n,k)$-admissible follows from the following observation made in \cite[Cor.\ 3.5(4)]{PPPP}: if we denote for
$\um\in\cS^{(n,k)}_i$ by $\um'$ the vector $\um$ with its middle entry $m_{i+1}$ replaced by $m'_{i+1}$, then
$\um'\in\cS^{(n,k)}_i$ when $k<m'_{i+1}<n-k+1$. According to the formula (\ref{eq:k_i_itoh_rest}) for $K_i$ this
means that when some term of $K_i$ contains a variable $x_j$ with $k+1\leqs j\leqs n-k$, then it contains also a
similar term with $x_j$ replaced by any $x_l$ with $k+1\leqs l\leqs n-k$. Considering the sum of these
substitutions yields a polynomial which is divisible by $L$. Therefore, $K_i$ is $(n,k)$-admissible. Let us show
that if for some $p\geqs1$, $K_i^{(p-1)}$ is $(n,k)$-admissible, then so is $K_i^{(p)}=e^{\Dp p}K_i^{(p-1)}$. We can
write $K_i^{(p-1)}=LK'+K''$, where $K'$ and $K''$ are independent of $x_{k+1},\dots,x_{n-k}$, and $\Dp p$
differentiates with respect to the variables $x_{k+1-\fr p2}$ and $x_{n-\fr{p-1}2}$. When $p\neq1$ and $p\neq2k+1$,
these variables are outside the range $k+1,...,n-k$, hence $\Dp {p}K_i^{(p-1)}=L\Dp pK'+\Dp pK''$, with $\Dp pK'$ and
$\Dp pK''$ independent of $x_{k+1},\dots,x_{n-k}$, so that $K_i^{(p)}$ is $(n,k)$-admissible. For $p=1$,
\begin{eqnarray*}%
  K_i^{(1)}&=&e^{\Dp 1}K_i^{(0)}=K_i^{(0)}+b_1\frac{\p^2}{\p x_{k+1} \p x_{n}}(LK'+K'')\\ &=&
  K_i^{(0)}+L\Dp1K'+\Dp1 K''+b_1\pp{K'}{x_n}\;,
\end{eqnarray*}%
showing that $ K_i^{(1)}$ is also $(n,k)$-admissible. The proof for $p=2k+1$ is very similar, since $\Dp{2k+1}$
differentiates with respect to the variables $x_1$ and $x_{n-k}$.
\end{proof}
We are now ready to show that every polynomial integral is in involution with every rational integral.

\begin{prop}\label{prp:all_in_invol}
  For $\ell=1,2,\dots,\fr{n-1}2-k$ and $i=0,1,\dots,k$,
  $$\pb{F_\ell^{\nkb},K_i^\nkb}^\nk_b=0\;.$$
\end{prop}
\begin{proof}
In view of Proposition \ref{prp:K_i_admissible}, we can write $K_i^b=LK'+K''$ where $K'$ and $K''$ are independent
of $x_{k+1},\dots,x_{n-k}$. Using Lemma \ref{lma:F_ders} twice,
\begin{equation*}
  \pb{F_\ell^b,K_i^b}_b=\pb{F_\ell^b,LK'+K''}_b=\pb{F_\ell^b,LK'}_b=K'\pb{F_\ell^b,L}_b\;.
\end{equation*}%
The Hamiltonian $H$ is of course also $(n,k)$-admissible, $H=L+H''$, with $H''$ independent of the variables
$x_{k+1},\dots,x_{n-k}$. Using that $F_\ell^b$ is a first integral (Theorem \ref{thm:defo_rat_invol}) and Lemma
\ref{lma:F_ders}, we can conclude that
\begin{equation*}
  \pb{F_\ell^b,K_i^b}_b=K'\pb{F_\ell^b,H-H''}_b=-K'\pb{F_\ell^b,H''}_b=0\;.
\end{equation*}%
\end{proof}
Combining the results obtained in this section, we can state and prove the main theorem on the integrability of the
systems $\LVB nk$, with $n\geqs 2k+1$. We denote, in that order,  by $H^{(n,k),b}_1,H^{(n,k),b}_2,$
$\dots,H^{(n,k),b}_{n-k-2}$ the following first integrals:
\begin{align*}
  F_1^b=G_1^b,F_2^b,\dots,F_{p-1}^b,G_2^b,\dots,G_{p-1}^b,F_p^b=G_{p}^b,\mbox{ when $n-k$ is even,}\\
              F_1^b,\dots,F_{p-1}^b,G_1^b,\dots,G_{p-1}^b,F_p^b=G_{p}^b,\mbox{ when $n-k$ is odd,}
\end{align*}
where $p:=\left[\frac{n-k}2\right]$.

\begin{thm}\label{thm:main}
  Consider the system $\LVB nk$, where $n\geqs 2k+1$.
  \begin{enumerate}
  \item[(1)] When $n>2k+1$, $\LVB nk$ is non-commutative integrable of rank $k+1$, with first integrals
    \begin{equation*}
      \qquad\qquad H=K_0^\nkb,K_1^\nkb\dots,K_k^\nkb,H^\nkb_1,H^\nkb_2,\dots, H^\nkb_{n-2k-2}\;.
    \end{equation*}
    The first $k+1$ functions of this list have independent Hamiltonian vector fields and are in involution with
    every function of the complete list.
  \item[(2)] $\LVB nk$  is Liouville  integrable with first integrals
    $$\qquad\quad H=K_0^\nkb,K_1^\nkb\dots,K_k^\nkb,H^\nkb_1,H^\nkb_2,
    \dots,H^\nkb_{s-1}\;,$$ where $s:=\left[\frac{n+1}2\right]-k.$
  \end{enumerate}
\end{thm}
\begin{proof}
We first consider (1). We have already checked the first item of Definition \ref{def:non-com}, namely that the
$k+1$ polynomials $K_i$ are first integrals of $\LV nk$, and are in involution with both the polynomial and
rational first integrals (Section \ref{par:pol_int} and Proposition \ref{prp:all_in_invol}). We need to check the second
item which says that the differentials of these first integrals are independent on a dense open subset of $\bbR^n$,
and similarly for the Hamiltonian vector fields associated to the polynomial first integrals. To do this, we use the fact
that the undeformed first integrals have this property, as they define a non-commutative integrable system of rank
$k+1$ (see \cite[Theorem 1.1]{PPPP}). Since all first integrals are rational functions and since the Poisson structure is
polynomial, it suffices to prove that the differentials (resp.\ vector fields) are independent at some point. The
argument is the same as the one used in Section \ref{par:rat_integrability} to derive the independence of the
first integrals of $\LVB n0$ from the independence of the first integrals of $\LV n0$: since the property is true at some point
$P$ when all parameters are zero, it is still true on a neighborhood of $P$ for small values of the parameters; by
rescaling the variables and parameters, one finds that at $P$ the property is true for all values of the
parameters. This proves (1). We now consider (2), the Liouville integrability. Since the rank of $\Pi_b$ is $n$
when $n$ is even and $n-1$ when $n$ is odd, we need $n/2$ independent first integrals in involution when $n$ is even and
$\frac{n+1}2$ when $n$ is odd. Clearly, the above list in (2) contains $k+s=\frac{n+1}2$ functions, which is the
right number, we know that they are pairwise in involution, and by the above argument they are independent. So they
define a Liouville integrable system.
\end{proof}
Item (1) in the theorem takes a slightly different form when $n=2k+1$. The constructed first integrals are then
polynomial and they define a non-commutative integrable system of rank $k$, which is equivalent to saying that it
is Liouville integrable, which is stated in (2), and was already proven in \cite{RCD}. The reason of this drop
in the rank of the non-commutative integrability when $n=2k+1$ is because, even though we have
$k+1$ polynomial first integrals that are in involution with all first integrals, like the general case of the $\LVB nk$
systems, now one of these $k+1$ polynomial first integrals is a Casimir and in order to establish the condition (2)
of Definition \ref{def:non-com} one has to exclude the Casimir from our set of first integrals.

\subsection{Examples}\label{par:examples}
For explicitness, we give below two examples, $\LVB 41$, which is the smallest new system with $k>0$ and $\LVB
71$, where one can see some non-trivial examples of the first integrals which we consider.

\begin{example}
For $n=4$ and $k=1$ the matrices $A^{(4,1)}$ and $B^{(4,1)}$ are given by
\begin{equation*}
  A^{(4,1)}=
  \begin{pmatrix}
    0&1&1&-1\\
    -1&0&1&1\\
    -1&-1&0&1\\
    1&-1&-1&0
  \end{pmatrix}\;,\quad
  B^{(4,1)}=
  \begin{pmatrix}
    0&0&b_3&-b_2\\
    0&0&0&b_1\\
    -b_3&0&0&0\\
    b_2&-b_1&0&0
  \end{pmatrix}\;.
\end{equation*}%
The corresponding system $\LVB 41$ is given by
\begin{eqnarray*}
  \dot x_1&=&x_1(x_2+x_3-x_4)+b_3-b_2\;,\\
  \dot x_2&=&x_2(-x_1+x_3+x_4)+b_1\;,\\
  \dot x_3&=&x_3(-x_1-x_2+x_4)-b_3\;,\\
  \dot x_4&=&x_4(x_1-x_2-x_3)+b_2-b_1\;.
\end{eqnarray*}
Besides the Hamiltonian $H=x_1+x_2+x_3+x_4$ it has an additional polynomial integral
$K_1^b=(x_1x_4+b_2)(x_2+x_3)+b_3x_4+b_1x_1$ which is easily seen to be a $(4,1)$-admissible polynomial.  The above
two polynomials give the Liouville integrability of the system $\LVB 41$ which coincides in this case with the
non-commutative integrability of rank $k+1=2$ just like in all $\LVB{2k+2}{k}$ systems.
\end{example}

\begin{example}
We now consider the case $n=7$ with $k=1$. The matrix $A:=A^{(7,1)}$ is the skew-symmetric Toeplitz matrix with
first line $(0,1,1,1,1,1,-1)$ and $B:=B^{(7,1)}$ is the skew-symmetric matrix whose only non-zero upper triangular
entries are $b_{1,6}=b_3$, $b_{1,7}=-b_2$ and $b_{2,7}=b_1$.
%
%
The corresponding system $\LVB 71$ is given by the equations
\begin{equation*}
\dot x_i = \sum_{j=1}^7\left(A_{i,j}x_ix_j+b_{i,j}\right),
\quad \text{for} \quad i=1,2,\ldots, 7.
\end{equation*}
Besides the Hamiltonian $H=x_1+x_2+\cdots+x_7$, the system $\LVB 71$ has one more independent polynomial first
integral $K_1^b$, given by
$$
  K_1^b=(x_2+x_3+x_4+x_5+x_6)(x_1x_7+b_2)+b_3x_7+b_1x_6\,,
$$
which is a $(7,1)$-admissible polynomial. It has also three rational first integrals given by
\begin{gather*}
  F_1^b=\frac{(x_1x_7+b_2)x_2x_4x_6+b_3x_2x_4x_7+b_1x_1x_4x_6+b_3b_1x_4}{x_3x_5},\\
  F_2^b=\frac{(x_1x_7+b_2)(x_2+x_3+x_4)x_6+b_3(x_2+x_3+x_4)x_7+b_1(x_1x_6+b_3)}{x_5},
\end{gather*}
and $G_2^b=\imath^* F_2^b$. The rank of the Poisson structure $\Pi^{(7,1)}_b$ is 6 and $F_1^b$ is a Casimir, invariant
under $\imath^*$. It can be seen that the above first integrals are obtained from the undeformed ones (obtained by
setting the parameters equal to zero), by applying on them the operator $e^{D_b}$ which now becomes
$$
e^{D_b}
=\left(I+b_3\frac{\p^2}{\p x_1\p x_6}\right)
\left(I+b_2\frac{\p^2}{\p x_1\p x_7}\right)
\left(I+b_1\frac{\p^2}{\p x_2\p x_7}\right).
$$
%
The system $\LVB 71$ is non-commutative integrable of rank $2$ with first integrals $H, K_1^b, F_1^b, F_2^b, G_2^b$ and is
also Liouville integrable with first integrals $H, K_1^b, F_1^b, F_2^b$ or $H, K_1^b, F_1^b, G_2^b$.
\end{example}


\section{Discretization of $\LVB n0$}\label{sec:discretizations}

In this section we construct a family of discretizations of $\LVB n0$. They are obtained from a discrete zero
curvature condition, which is the compatibility condition of a linear system $L \Psi=\lambda \Psi, \tilde \Psi=N
\Psi,$ where $L$ is the Lax matrix of $\LV n0$, which appears in \eqref{eq:bogo_deformed_lax}. We prove that an
important class of these discretizations, which includes the Kahan (also called Kahan-Hirota-Kimura) discretization
of $\LV n0$ has the following integrability properties: it has the rational first integrals of $\LV n0$ as invariants,
and so it is both Liouville and superintegrable; also it has an invariant measure.

Throughout this section, $(n,k)=(n,0)$ is fixed and so we will drop $(n,0)$ from the notation for the invariants,
the Poisson structure, and so on. Also, since we have in this case only one parameter $b_{1,n}$, we will
denote it by~$\beta$, as we did in Section \ref{sec:rational}.

\subsection{Preliminaries}
We first recall a few basic definitions and properties of discrete maps and their integrability.  By a
\emph{discrete map} of $\bbR^n$ we mean an algebra homomorphism $\Phi:\bbR(x_1, x_2, \ldots, x_n)\rightarrow
\bbR(x_1, x_2, \ldots, x_n)$, where $x_1,\dots,x_n$ are as elsewhere in this paper the natural coordinates on
$\bbR^n$. Such a map is the pullback of a unique rational map $\phi:\bbR^n\to\bbR^n$, i.e., for any rational
function $F$, one has $\Phi(F)=\phi^*(F)=F\circ\phi$. We will also use the convenient abbreviations $\tilde F$ for
$\Phi(F)$. Similarly, for a matrix $P=(p_{i,j})$ whose entries are rational functions of $\bbR^n$, we will write
$\tilde{P}$ for the matrix~$(\tilde p_{i,j})$.

When $\bbR^n$ is equipped with a Poisson structure $\Pi$, then saying that $\Phi$ is a homomorphism of Poisson
algebras is tantamount to saying that $\phi$ is a Poisson map; we will simply say that $\Phi$ \emph{preserves} the
Poisson structure $\Pi$. Also, on $\bbR^n$ we have a natural $n$-form, $\diff x_1\wedge\dots\wedge \diff x_n$ which
allows us to identify rational measures with rational $n$-forms and with rational functions. We will say that
$\Phi$ is \emph{measure preserving}, with preserved measure $F$, if it preserves the $n$-form
$F\mathrm{d}x_1\wedge\mathrm{d}x_2\wedge\ldots\wedge\mathrm{d}x_n$ in the sense that
$$
  F\mathrm{d}x_1\wedge\mathrm{d}x_2\wedge\ldots\wedge\mathrm{d}x_n=
  \tilde F\mathrm{d}\tilde x_1\wedge\mathrm{d}\tilde x_2\wedge\ldots\wedge\mathrm{d}\tilde x_n\,.
$$
A rational function $F$ is called an \emph{invariant} of $\Phi$ if $\tilde F=F$.  We also recall the definition of
an integrable map \cite{Veselov_maps}.
\begin{defn} Suppose that $\Phi$ is a discrete map of $\bbR^n$.
  \begin{enumerate}
    \item[(1)] $\Phi$ is \emph{Liouville integrable} if there exist $n-r$ functionally independent
      invariants of $\Phi$, which are in involution with respect to a Poisson structure $\Pi$, where $r$ is half
      the rank of $\Pi$.
    \item[(2)] $\Phi$ is \emph{superintegrable} if it has $n-1$ functionally independent invariants and is
      measure preserving.
  \end{enumerate}
\end{defn}

\subsection{Discrete maps from a linear problem}
Recall that $\LVB n0$ is by definition a reduction of $\LVB {2m+1}m$, with $m:=n-1$,
and that
\begin{equation}\label{eq:L_matrix}
  L=X+\lambda^{-1}\Delta+\lambda M
\end{equation}
is the square matrix of size $2m+1$, where
\begin{equation}\label{eq:L_matrix_disc}
  X_{i,j}:=x_i\delta_{i,j+m}\;,\quad \Delta_{i,j}:=b_{i+m,j}\delta_{i,j}\;,
  \quad M_{i,j}:=\delta_{i+1,j}\;,
\end{equation}
and the indices of $x, b$ and $\delta$ are considered modulo $2m+1$.
We consider the compatibility conditions
of the linear system
\begin{equation}\label{eq:lin_sys}
  L \Psi=\lambda \Psi, \tilde \Psi=N \Psi\,,
\end{equation}
where $L$ is the Lax matrix of $\LVB {2m+1}{m}$, recalled above, and $\Psi$ is an $n$-dimensional
vector whose entries are rational functions on $\mathbb R^{2m+1}$. Recall that $\tilde{\Psi}$ is the
vector $\Psi$ with $\Phi$ applied to its entries. The $(2m+1)\times (2m+1)$ matrix $N$ is defined as
\begin{equation}\label{eq:N_matrix}
  N=D-\lambda K\,,
\end{equation}
where $K_{i,j}:=\delta_{i,j+m}$ and $D_{i,j}:=d_i\delta_{i,j}$ for some functions $d_i$ that will be determined
from the compatibility condition of \eqref{eq:lin_sys}, which reads $\tilde{L}N=NL$. Since $N$ is invertible, it
means that $\tilde{L}=NLN^{-1}$ and therefore the coefficients of the characteristic polynomial of $L$, which are
rational functions on $\bbR^{2m+1}$, are invariants of $\Phi$.  The above ansatz for $N$ was taken so that
$NLN^{-1}$ equals $L$ at the entries with constant values. Therefore, the compatibility condition $\tilde{L}N=NL$
reduces to a system of equations for the $\tilde{x}_i$ and $d_i$ variables, which we make explicit in the following
proposition:
\begin{prop}
  The compatibility condition $\tilde{L}N=NL$ of the linear system \eqref{eq:lin_sys} is equivalent to the
  following system of equations:
  \begin{equation}\label{eq:sys_eq_map}
    d_{i+1}-d_i+x_{m+1+i}-\tilde{x}_i=0 \ \text{and} \quad
    d_{m+1+i}\tilde{x}_i-d_ix_i+b_{i,m+1+i}-b_{m+i,i}=0\;,
  \end{equation}
  for $i=1,2,\ldots, 2m+1=2n-1$.
\end{prop}
\begin{proof}

Notice first that the equations we get from the first line of $\tilde{L}N=NL$ are
$$
d_2-d_1+x_{m+2}-\tilde{x}_1=0 \quad\hbox{and}\quad d_{m+2}\tilde{x}_1-d_1x_1+b_{1,m+2}-b_{m+1,1}=0\;,
$$
which is (\ref{eq:sys_eq_map}) for $i=1$. Because of the form of the matrices appearing in
$\tilde{L}N=NL$ the other equations are obtained by shifting all indices by $1$, and the result follows.
\end{proof}


We now reduce these equations to $\LVB n0$, by setting $x_i=\tilde x_i=0$ for $i=n+1, n+2, \ldots, 2m+1$ and
$b_{i,m+i}=0$ for $i=2,3,\ldots,2m+1$, where we recall that $m=n-1$ and that we denote the single parameter
$b_{1,n}$ of $\LV n0$ as $b_{1,n}=\beta$.  The system \eqref{eq:sys_eq_map} is then transformed to the following
one:
\begin{equation}\label{eq:sys_comp}
  \begin{aligned}[c]
    \tilde x_i&=d_{i+1}-d_{i}\;,\\
    \tilde x_n&=x_1+d_{n+1}-d_{n}\;,\\
    x_{i+1}&=d_{n+i}-d_{n+i+1}\;,\\
    d_{n+i+1}\tilde x_{i+1}&=d_{i+1}x_{i+1}\;,\\
    d_{n+1}\tilde x_1&=d_1x_1-\beta\;,\\
    d_{1}\tilde x_n&=d_nx_n+\beta\,,
  \end{aligned}
  \qquad
  \begin{aligned}[c]
    i=1,\ldots, n-1,\\
    \\
    i=1,\ldots, n-1,\\
    i=1,\ldots, n-2,\\
    \\
    \\
  \end{aligned}
\end{equation}
where the first three equation are instances of the first equation in \eqref{eq:sys_eq_map} and the last three
equations of the second one.  Before solving the above system, we recall from Section \ref{par:solutions} the
alternative coordinates $u_1,\dots,u_n$ for $\bbR^n$, in which the system $\LVB n0$ completely separates. They are
defined by $u_i=\sum_{j=1}^ix_j$ for all $i=0,1,2,\ldots,n$.  For $i=n, \, u_n$ is just the Hamiltonian,
$u_n=H=x_1+x_2+\ldots+x_n$.
\begin{prop}\label{prp:comp_sol}
  For any rational function $\cR\in\bbR(x_1,\dots,x_n)$, different from the $n$ functions $u_i-H$, with
  $i=1,\dots,n$, the reduced compatibility equations (\ref{eq:sys_comp}) have a unique solution for $\tilde
  x_1,\dots,\tilde x_n$ and for $d_2,\dots,d_{2n-1}$, with $d_1=\cR$. It is given by
  \begin{gather}
    \begin{split}\label{eq:sol_comp}
      \tilde{x}_i&=d_{i+1}-d_i\;, \qquad\quad\ i=1,2,\ldots, n-1\;, \\
      \tilde{x}_n&=x_1+d_{n+1}-d_n\;,\\
      d_i&=\frac{\cR(\cR+H)-\beta} {\cR+H-u_{i-1}}\;,\ \quad i=2,3,\ldots, n\;,\\
      d_{n+i}&=\cR+H-u_i\;,\qquad\quad i=1,2,\ldots, n-1\;.
    \end{split}
  \end{gather}
\end{prop}
\begin{proof}
We first show how the third and fourth equations in (\ref{eq:sol_comp}) are derived from (\ref{eq:sys_comp}). The
last equation is obtained from the third equation in (\ref{eq:sys_comp}): for $i=1,\dots,n-1$,
\begin{equation*}
  H-u_i=\sum_{j=i+1}^n x_j=\sum_{j=i+1}^n(d_{n+j-1}-d_{n+j})=d_{n+i}-\cR\;,
\end{equation*}%
where we have used that, by periodicity, $d_{2n}=d_{2m+2}=d_1=\cR $. In order to derive the third equation in
(\ref{eq:sol_comp}), one first uses the first three equations in (\ref{eq:sys_comp}) to substitute $\tilde x_i$
($i=1,\ldots,n$) and $x_i$ ($i=2,3,\ldots, n$) in the fourth and fifth equations in \eqref{eq:sys_comp}, to obtain,
in that order,
\begin{gather}\label{sys00}
  \begin{split}
    d_{i+1}d_{n+i}&=d_{i+2}d_{n+1+i}\;, \quad i=1,\ldots, n-2\;,\\
    \cR (x_1+d_{n+1})&=d_2d_{n+1}+\beta\;.
\end{split}
\end{gather}
The first equation in (\ref{sys00}) says that $d_{i+1}d_{n+i}$ is independent of $i$ for $i=1,\dots,n-2$, while the
second equation says that this constant value is equal to $\cR (x_1+d_{n+1})-\b,$
\begin{equation}\label{eq:int}
  d_{i+1}d_{n+i}=\cR (x_1+d_{n+1})-\b=\cR (\cR +H)-\b\;,
\end{equation}%
for $i=1,\dots,n-2$. By our assumption on $\cR $, the $d_{n+i}=\cR +H-u_i$ with $i=1,\dots,n-1$ are all different from
zero, so that we can divide (\ref{eq:int}) by~$d_{n+i}$. It yields the third equation in (\ref{eq:sol_comp}). This
shows that (\ref{eq:sol_comp}) is the only possible solution for (\ref{eq:sys_comp}) with $d_1=\cR $. That it is
indeed a solution is easily verified by substituting the formulas (\ref{eq:sol_comp}) in (\ref{eq:sys_comp}).
\end{proof}
We now define a discrete map using the solution of \eqref{eq:sys_comp} given by \eqref{eq:sol_comp}.  Let
$\cR$ be a rational function, with $\cR \neq u_{i}-H$ for all $i=1,2,\ldots, n-1$, and let $\Phi_\cR$ be the discrete
map $x_i\mapsto \tilde x_i$, defined by the formulas
\begin{gather}
\begin{split}
\label{eq:phi_expl_x}
\tilde{x}_1&=\frac{x_1\cR-\beta}{\cR +H-x_1}\;,\\
\tilde{x}_i&= x_i\frac{\cR (\cR +H)-\beta}{(\cR +H-u_i)(\cR +H-u_{i-1})}\;,\quad i=2, 3, \ldots, n-1\;,\\
\tilde{x}_n&= \frac{x_n(\cR +H)+\beta}{\cR +x_n}\;.
\end{split}
\end{gather}
Using the first $n$ equations in \eqref{eq:sys_comp}, we get $\tilde{u}_i=d_{i+1}-\cR $
and therefore the map is given in terms of the coordinates $u_i$ by $\tilde{u}_n=u_n$ and
\begin{equation}
\label{eq:phi_expl_u}
\tilde{u}_i=\frac{u_i \cR -\beta}{\cR +H-u_i}, \quad i=1,2, \ldots, n-1\;.
\end{equation}
Since the discrete map $\Phi_\cR$ is by construction isospectral, it has the coefficients of the
characteristic polynomial of the Lax matrix $L$ as invariants. However, as we noted just after equation
\eqref{eq:K_are_defos}, we get in this way only one invariant, namely the Hamiltonian $H$.
We show in the next proposition that  $\Phi_\cR $ also preserves the rational first integrals
of $\LVB n0$. By $\phi_R$ we will denote the rational map underlying $\Phi_R$.

\begin{prop}\label{prp:rat_preser}
Let $P$ be any point of $\bbR^n$ for which $Q:=\phi_\cR(P)$ is defined. Then $Q$ belongs to the integral
curve of the continuous system $\LV n0$ starting at $P$. In particular, the discrete map $\Phi_\cR $ preserves all
the first integrals of $\LVB n0$.
\end{prop}
\begin{proof}
Since $\Phi_\cR$ is a discrete map and the first integrals of $\LVB n0$ are rational functions,
it suffices to show that for a generic $P$ of $\bbR^n$ for which $Q:=\phi_\cR(P)$ is defined,
$Q$ belongs to the integral curve of the continuous system $\LVB n0$ starting at $P$.

We use the notation of the proof of Proposition \ref{prp:solutions}: we denote by $\gamma_P$ the integral
curve of (\ref{eq:LV0_u}) starting from $P$, and we write $u_i(t)=u_i(\gamma_P(t))$.
We denote by $h$ the value of the Hamiltonian at $P$ and by $\Delta_0$
a square root of ${h^2+4\beta}$, which may be real or imaginary.
Also, let $r_0$ denote the value of $\cR $ evaluated at $P$ and
$Q=(Q_1,Q_2, \ldots, Q_n)$.

It is clear from the above that we only need to show that for each $P$ such that $Q$ is defined there exists a $t$,
depending only on $P$, such that $Q_i=u_i(t)$ for all $i=1,2,\ldots, n-1$.  It is also clear that we may consider
our system $\LVB n0$ living on $\bbC^n$ and therefore the integral curves are defined on all of $\bbC$ minus a
discrete set (see Section \ref{par:solutions} for details and comments).

We only need to consider the case that $\Delta_0\neq 0$. In this case the solution of $\LVB n0$, for
$$
  t=-\frac{\ln\left(\frac{h+2r_0+\Delta_0}{h+2r_0-\Delta_0}\right)}{\Delta_0}\,,
$$
gives $Q_i=u_i(t)$ for all $i$, as can be seen by comparing the formulas \eqref{eq:phi_expl_u}
and the explicit solution of $\LVB n0$ given in Proposition \ref{prp:solutions}.
\end{proof}
\subsection{Integrable discretization of $\LVB n0$}
For a general rational function $\cR$, the discrete map $\Phi_\cR =\phi_\cR^*$ cannot be expected to have any integrability
properties. We establish in this subsection a few results under the assumption that $\cR$ is a first integral of
$\LV n0$, or under the stronger hypothesis that $\cR$ depends on $H$ only. We first prove that, under these
conditions, $\Phi_\cR$ is birational.
\begin{prop}\label{prp:bir}
  Suppose that $\cR$ is a first integral of $\LVB n0$. Then $\phi_\cR $ is a birational map, so that $\Phi_\cR $
  is an algebra automorphism of $\mathbb R(x_1,\dots,x_n)$.
\end{prop}
\begin{proof}
Let $\cR $ be as announced, so that $\tilde \cR =\cR $, in view of Proposition \ref{prp:rat_preser}.
If we exchange $x_i$ and $\tilde x_{n+1-i}$ as well as $d_i$ and $d_{2n+1-i}$ in the equations
\eqref{eq:sys_comp} we get the same set of equations: the first and third equations are permuted, as well as the
last ones, while the other two are unchanged. Since we know from Proposition (\ref{prp:comp_sol}) that given
$d_1:=\cR $ the reduced compatibility equations (\ref{eq:sys_comp}) have a unique solution for $\tilde
x_1,\dots,\tilde x_n$ and for $d_2,\dots,d_{2n-1}$, in terms of $x_1,\dots,x_n$, this means given $d_1:=\tilde \cR
=\cR $, they also have a unique solution for $x_1,\dots,x_n$ and for $d_2,\dots,d_{2n-1}$, in terms of $\tilde
x_1,\dots,\tilde x_n$. Therefore, the map $\phi_\cR $, defined by the solutions of the system (\ref{eq:sys_comp}),
is birational and so $\Phi_\cR $ is an algebra automorphism of $\mathbb R(x_1,\dots,x_n)$.
\end{proof}
We now prove that $\Phi_\cR $ is, under the same assumption on $\cR $, measure preserving.
\begin{prop} \label{prp:prs_form}
Suppose that $\cR $ is a first integral of $\LVB n0$. Then the discrete map $\Phi_\cR $ preserves the
rational $n$-form
$$
  \Omega_b:=\frac{\mathrm d x_1 \wedge \mathrm d x_2 \wedge \dots
\wedge \mathrm d x_n}{x_1x_2\ldots x_n +\beta x_2x_3\ldots
  x_{n-1}}\;.
$$
\end{prop}
\begin{proof}
We need to show that
$$
  \frac{\mathrm d x_1 \wedge \mathrm d x_2 \wedge \dots \wedge \mathrm
d x_n}{x_1x_2\ldots x_n+\beta x_2x_3\ldots
  x_{n-1}}= \frac{\mathrm d \tilde{x}_1 \wedge \mathrm d \tilde{x}_2
\wedge \dots \wedge \mathrm d \tilde{x}_n}
  {\tilde{x}_1\tilde{x}_2\ldots \tilde{x}_n
+\beta\tilde{x}_2\tilde{x}_3\ldots \tilde{x}_{n-1}}\;.
$$
Since the coordinate change between the coordinates $u_i$ and $x_i$ have
triangular form, and since the functions
$\tilde u_i$ depend in the same way on the $\tilde x_i$, i.e., $\tilde
u_i=\sum_{j=1}^i\tilde x_j$, we have that
$$
\left|\frac{\p (\tilde{x}_1,\tilde{x}_2,\ldots,\tilde{x}_n)}
{\p (\tilde{u}_1,\tilde{u}_2,\ldots,\tilde{u}_n)}\right|=
\left|\frac{\p (x_1,x_2,\ldots,x_n)}{\p (u_1,u_2,\ldots,u_n)}\right|=1\;,
$$
where the above determinants are the Jacobian determinants of these two
transformations. This implies that we need
to show that
$$
  \frac{\mathrm d u_1 \wedge \mathrm d u_2 \wedge \dots \wedge \mathrm
d u_n}{x_1x_2\ldots x_n +\beta
  x_2x_3\ldots x_{n-1}}= \frac{\mathrm d \tilde{u}_1 \wedge \mathrm d
\tilde{u}_2 \wedge \dots \wedge \mathrm d
  \tilde{u}_n} {\tilde{x}_1\tilde{x}_2\ldots \tilde{x}_n
+\beta\tilde{x}_2\tilde{x}_3\ldots \tilde{x}_{n-1}}\;.
$$
We assume for the moment that $\cR $ is any rational function such that $\Phi_\cR$ is well
defined and we denote by $\cR _j$ the partial derivatives $\frac{\p
\cR }{\p u_j}$ and by $\tilde u_{i,j}$ the partial derivatives $\frac{\p \tilde{u}_i}{\p u_j}$.
The explicit formulas \eqref{eq:phi_expl_u} for $\tilde{u}_i$
and \eqref{eq:LV0_u} for $\dot{u}_i$ give,
for any $i=1,2,\ldots, n-1$ and any $j\notin\set{i,n}$, that
\begin{equation}\label{eq:part_tilu}
  \tilde u_{i,j}=\frac{\cR _j\dot{u}_i}{(\cR+H-u_i)^2}\;, \quad
  \tilde u_{i,i}=\frac{\cR _i\dot{u}_i+\cR H+\cR ^2-\beta}{(\cR +H-u_i)^2}\;,
\end{equation}
and $\tilde u_{n,j}=\delta_{n,j}$, since $\tilde u_n = u_n$. We do not
display the formulas for the derivatives $\tilde u_{i,n}$ ($i=1,\dots,n-1$)
since we do not need them in this proof.
Differentiating \eqref{eq:phi_expl_u}
and rearranging, we get
\begin{gather*}
  \mathrm{d}\tilde{u}_i=\sum_{j=1}^n\tilde u_{i,j}\mathrm du_j\;,\quad i=1,2,
\ldots, n\;.
\end{gather*}
Taking the wedge product of the above $n$ equations we obtain
$$
  \mathrm d \tilde{u}_1 \wedge \mathrm d \tilde{u}_2 \wedge \dots \wedge \mathrm d \tilde{u}_n=  \det(U)
  \ \mathrm d u_1 \wedge \mathrm d u_2 \wedge \dots \wedge \mathrm d u_n\;,
$$
where $U=(\tilde u_{i,j})$. The last line of $U$ is the vector $(0,0,\ldots, 1)$ and therefore expanding the determinant
$\det(U)$ with respect to the last line we get $\det(U)=\det(V)$ where $V$ is the minor of $U$ obtained by removing
its last row and last column.  According to the first formula in \eqref{eq:part_tilu} the determinant of the
$(n-1)\times(n-1)$ matrix $V$ has the following form:
$$
V=
\frac{1}{S}
\det\begin{pmatrix}
\cR _1\dot{u}_1 + R & \cR _2\dot{u}_2 & \cR _3\dot{u}_3 &\cdots & \cR _n \dot{u}_{n-1}\\
\cR _1\dot{u}_1& \cR _2\dot{u}_2+R & \cR _3\dot{u}_3 &\cdots & \cR _n \dot{u}_{n-1}\\
\cR _1\dot{u}_1& \cR _2\dot{u}_2 & \cR _3\dot{u}_3+R &\cdots & \cR _n \dot{u}_{n-1}\\
 &\vdots & & \ddots & \vdots\\
\cR _1\dot{u}_1& \cR _2\dot{u}_2 & \cR _3\dot{u}_3 &\cdots & \cR _n \dot{u}_{n-1}+R\\
\end{pmatrix}\;,
$$
where $S=\prod_{j=1}^{n-1}(\cR +H+u_i)^2$ and $R=\cR H+\cR ^2-\beta$. The above matrix is written as $W+RI_{n-1}$,
where $I_{n-1}$ is the $(n-1)\times (n-1)$ identity matrix and $W$ has $n$ equal lines. This means that $W$ has
only two eigenvalues $\lambda_1$ and $\lambda_2$. The first one is $\lambda_1=\sum_{j=1}^{n-1}\cR _j\dot{u}_j,$
which is of multiplicity $1$ and the other one is $\lambda_2=0$ of multiplicity $n-2$.  In the particular case
where $\cR $ is a first integral of $\LVB n0$, the eigenvalue $\lambda_1$ reduces to zero (since $\dot{u}_n=0$).
This shows that, in that case,
$$
  \det(V)=\frac{R^{n-1}}{S}=\frac{(\cR H+\cR ^2-\beta)^{n-1}}{\prod_{i=1}^{n-1}(\cR +H-u_i)^2}\;.
$$
Therefore, what we need to show is that
$$
  \prod_{i=1}^{n-1}\frac{\cR H+\cR ^2-\beta}{(\cR +H-u_i)^2}=
\frac{\tilde{x}_1\tilde{x}_2\ldots
  \tilde{x}_n +\beta\tilde{x}_2\tilde{x}_3\ldots \tilde{x}_{n-1}}
{x_1x_2\ldots x_n +\beta x_2x_3\ldots x_{n-1}}\;.
$$
A comparison with the explicit formulas \eqref{eq:phi_expl_x} gives that
$$
\frac{x_2x_3\ldots x_{n-1}}{\tilde{x}_2\tilde{x}_3\ldots \tilde{x}_{n-1}}
\frac{(\cR H+\cR ^2-\beta)^{n-1}}{\prod_{i=1}^{n-1}(\cR +H-u_i)^2}=
\frac{\cR H+\cR ^2-\beta}{(\cR +H-u_1)(\cR +H-u_{n-1})}\;.
$$
To complete the proof, it remains to be shown that
$$
\frac{\tilde{x}_1\tilde{x}_n+\beta}{x_1x_n+\beta}=\frac{\cR H+\cR ^2-\beta}{(\cR +H-u_1)(\cR +H-u_{n-1})}\;.
$$
This can be done by substituting the formulas for $\tilde x_1$ and $\tilde x_n$, given in \eqref{eq:phi_expl_x}.
\end{proof}

In order to preserve the Poisson structure, one needs stronger conditions on $\cR$, as given in the following
proposition:
\begin{prop}\label{prp:prs_pois_str}
  Suppose that $\cR$ is a rational function of the Hamiltonian~$H$. Then the map $\Phi_\cR$
  preserves the Poisson structure $\Pi_b$.
\end{prop}
\begin{proof}
We give the proof using the coordinates $u_i$ (see Section \ref{par:solutions}, in particular the formulas
(\ref{eq:pb_in_u}) for the Poisson structure in terms of these coordinates). According to (\ref{eq:pb_in_u}) we
need to show that $\pb{\tilde{u}_i,\tilde{u}_j}_b=\tilde{u}_i(\tilde{u}_j-\tilde{u}_i)$ and that
$\pb{\tilde{u}_\ell,\tilde{u}_n}_b=\tilde{u}_\ell(\tilde{u}_n-\tilde{u}_\ell)+\beta$ for $1\leqs i<j<n$ and $1\leqs\ell<n$.

According to \eqref{eq:phi_expl_u} $\tilde u_i$ depends only on $u_i$ and $H$,
so we give only the non-zero derivatives $\tilde{u}_{i,i}$ and $\tilde{u}_{i,n}$
for $i<n$. They are (for $\tilde{u}_{i,i}$ see the second formula in \eqref{eq:part_tilu}):
\begin{gather}
\label{eq:partial_tu}
\tilde u_{i,i}=\frac{\cR H+\cR^2-\beta}{(\cR+H-u_i)^2}\quad \text{and} \quad
\tilde u_{i,n}=\frac{\cR_Hu_iH-\cR_Hu_i^2+\cR_H\beta-\cR u_i+\beta}{(\cR+H-u_i)^2}\;,
\end{gather}
where $\cR_H=\frac{\mathrm{d}\cR}{\mathrm{d}H}$. We then have, for all $i,j$ and $\ell$ as above,
\begin{eqnarray*}
  \pb{\tilde{u}_i,\tilde{u}_j}_b&=&\tilde u_{i,i}\tilde u_{j,j}\pb{u_i,u_j}_b
  +\tilde u_{i,i}\tilde u_{j,n}\pb{u_i,u_n}_b-
  \tilde u_{i,n}\tilde u_{j,j}\pb{u_j,u_n}_b\\
    &=&\frac{(u_i-u_j)(\beta-\cR u_i)(\cR H+\cR^2-\beta)}{(\cR+H-u_i)^2(\cR+H-u_j)}=
  \tilde{u}_i(\tilde{u}_j-\tilde{u}_i)\;,\\
  \pb{\tilde{u}_\ell,H}_b&=&\frac{(\cR H+\cR^2-\beta)(Hu_\ell-u_\ell^2+\beta)^2}{(\cR+H-u_\ell)^2}=
    \tilde{u}_\ell(H-\tilde{u}_\ell)+\beta\;,
\end{eqnarray*}
which establishes the required equalities, since $\tilde u_n=u_n=H$.
\end{proof}

The above propositions \ref{prp:rat_preser}--\ref{prp:prs_pois_str} lead to the following theorem.
\begin{thm}\label{thm:phi_int}
  Let $\cR $ be a rational function, depending on the Hamiltonian $H$ only. Then the discrete map $\Phi_\cR $ of
  $\LVB n0$ has the following properties:
\begin{enumerate}
  \item[(1)] It is birational;
  \item[(2)] It preserves the Poisson structure $\Pi_b$;
  \item[(3)] It is measure preserving: it preserves the volume form $\Omega_b$;
  \item[(4)] It is Liouville integrable with $H$ and the rational functions $F_\ell^b$ as invariants;
  \item[(5)] It is superintegrable with  $H$ and the rational functions $F_\ell^b$ and $G_\ell^b$ as invariants.
\end{enumerate}
\end{thm}
Under the weaker hypothesis that $\cR$ depends only on the invariants of $\LVB n0$, items (1), (3) and (5) still
hold, but (2) and (4) may not hold.

\subsection{Kahan discretization of $\LVB n0$}
In this subsection we consider the Kahan discretization of the systems $\LVB n0$. We show that the Kahan map is of
the form $\Phi_\cR $, for a specific choice of the rational function $\cR $, depending on the Hamiltonian $H$ only,
and so all integrability properties that we have seen in Theorem \ref{thm:phi_int} hold for the Kahan map as well.

We first define the Kahan map for $\LVB n0$. Since the Kahan discretization commutes with any linear change of
variables, we can do the Kahan discretization in the $u_i$ coordinates, instead of the $x_i$ coordinates, i.e.,
apply it on the vector field (\ref{eq:LV0_u}). Following the recipe \cite{Rein_kah}, we obtain for the Kahan
discretization with step size $2\epsilon$ the following system of equations:
\begin{gather}\label{eq:kahan_impl_u}
  \bar{u}_i-u_i=\e u_i\left(H-\bar{u}_i\right)+\e\bar{u}_i\left(H-u_i\right)+2\beta, \quad i=1,2,\ldots, n-1,
\end{gather}
where we have used that $H=u_n$. Since $H$ is a linear first integral of $\LVB n0$, it is an invariant for the
Kahan map. The system \eqref{eq:kahan_impl_u} is diagonal with solution
\begin{equation}
\label{eq:kahan_expl_u}
\bar{u}_i=\frac{(1+\epsilon H)u_i+2\e\beta}{1-\e H+2\e u_i}\;, \quad i=1,2, \ldots, n-1\;,
\end{equation}
and $\bar{u}_n=u_n$. This defines the Kahan map. Comparing the formulas \eqref{eq:kahan_expl_u} and
\eqref{eq:phi_expl_u} it is clear that the Kahan map is of the form $\Phi_\cR$, with
\begin{equation}\label{eq:R_for_Kahan}
  \cR =-\frac{1+\e H}{2\e}\;.
\end{equation}%
Notice that $\cR$ depends on $H$ only. Therefore, we get by Theorem \ref{thm:phi_int} the following results on the
Kahan discretization of $\LVB n0$, which generalize the results on the integrability of the Kahan
discretization of $\LV n0$, which were first established in \cite{KKQTV}:
\begin{thm}\label{thm:kahan_int}
The Kahan map of $\LVB n0$ has the following properties:
\begin{enumerate}
  \item[(1)] It is birational;
  \item[(2)] It preserves the Poisson structure $\Pi_b$;
  \item[(3)] It is measure preserving: it preserves the volume form $\Omega_b$;
  \item[(4)] It is Liouville integrable with $H$ and the rational functions $F_\ell^b$ as invariants;
  \item[(5)] It is superintegrable with  $H$ and the rational functions $F_\ell^b$ and $G_\ell^b$ as invariants.
\end{enumerate}
\end{thm}
As a byproduct of our analysis, we find that the Kahan map of $\LVB n0$ arises as the compatibility conditions of a
linear system. It would be interesting to see if there are other examples where the Kahan map is of this form, as
it links the Kahan map to isospectrality, so it may have non-trivial applications to the study of the integrability
of the Kahan map of other integrable systems.


\begin{thebibliography}{10}

\bibitem{adlermoerbekevanhaecke2004}
M.~Adler, P.~van Moerbeke, and P.~Vanhaecke.
\newblock {\em Algebraic integrability, {P}ainlev\'e geometry and {L}ie
  algebras}, volume~47 of {\em Ergebnisse der Mathematik und ihrer
  Grenzgebiete. 3. Folge. A Series of Modern Surveys in Mathematics [Results in
  Mathematics and Related Areas. 3rd Series. A Series of Modern Surveys in
  Mathematics]}.
\newblock Springer-Verlag, Berlin, 2004.

\bibitem{Adler1993}
V.~Adler.
\newblock Recutting of polygons.
\newblock {\em Funct. Anal. Appl.}, 27(2):79--80, 1993.

\bibitem{Bog1}
O.~I. Bogoyavlenskij.
\newblock Some constructions of integrable dynamical systems.
\newblock {\em Izv. Akad. Nauk SSSR Ser. Mat.}, 51(4):737--766, 910, 1987.

\bibitem{Bog2}
O.~I. Bogoyavlenskij.
\newblock Integrable {L}otka-{V}olterra systems.
\newblock {\em Regul. Chaotic Dyn.}, 13(6):543--556, 2008.

\bibitem{Rein_kah}
E.~Celledoni, R.~I. McLachlan, D.~I. McLaren, B.~Owren, and G.~R.~W. Quispel.
\newblock Integrability properties of {K}ahan's method.
\newblock {\em Journal of Physics A: Mathematical and Theoretical},
  47(36):365202, aug 2014.

\bibitem{PPPP}
P.~A. Damianou, C.~A. Evripidou, P.~Kassotakis, and P.~Vanhaecke.
\newblock Integrable reductions of the {B}ogoyavlenskij-{I}toh
  {L}otka-{V}olterra systems.
\newblock {\em J. Math. Phys.}, 58(3):032704, 17, 2017.

\bibitem{RCD}
C.~A. Evripidou, P.~Kassotakis, and P.~Vanhaecke.
\newblock Integrable deformations of the {B}ogoyavlenskij-{I}toh
  {L}otka-{V}olterra systems.
\newblock {\em Regul. Chaotic Dyn.}, 22(6):721--739, 2017.

\bibitem{PPC}
C.~A. Evripidou, P.~H. van~der Kamp, and C.~Zhang.
\newblock {D}ressing the {D}ressing {C}hain.
\newblock {\em SIGMA}, 14:59--73, 2018.

\bibitem{FordyHone2014}
A.~Fordy and A.~Hone.
\newblock Discrete integrable systems and {P}oisson algebras from {C}luster
  maps.
\newblock {\em Commun. Math. Phys.}, 325:527--584, 2014.

\bibitem{HirotaKimura2000-1}
R.~Hirota and K.~Kimura.
\newblock Discretization of the {E}uler top.
\newblock {\em J. Phys. Soc. Jap.}, 69:627--630, 2000.

\bibitem{HirotaKimura2000-2}
R.~Hirota and K.~Kimura.
\newblock Discretization of the {L}agrange top.
\newblock {\em J. Phys. Soc. Jap.}, 69:3193--3199, 2000.

\bibitem{itoh1}
Y.~Itoh.
\newblock Integrals of a {L}otka-{V}olterra system of odd number of variables.
\newblock {\em Progr. Theoret. Phys.}, 78(3):507--510, 1987.

\bibitem{itoh2}
Y.~Itoh.
\newblock A combinatorial method for the vanishing of the {P}oisson brackets of
  an integrable {L}otka-{V}olterra system.
\newblock {\em J. Phys. A}, 42(2):025201, 11, 2009.

\bibitem{KQV}
T.~E. Kouloukas, G.~R.~W. Quispel, and P.~Vanhaecke.
\newblock Liouville integrability and superintegrability of a generalized
  {L}otka-{V}olterra system and its {K}ahan discretization.
\newblock {\em J. Phys. A}, 49(22):225201, 13, 2016.

\bibitem{PLV}
C.~Laurent-Gengoux, A.~Pichereau, and P.~Vanhaecke.
\newblock {\em Poisson structures}, volume 347 of {\em Grundlehren der
  Mathematischen Wissenschaften [Fundamental Principles of Mathematical
  Sciences]}.
\newblock Springer, Heidelberg, 2013.

\bibitem{Lotka}
A.~J. Lotka.
\newblock {\em Analytical theory of biological populations}.
\newblock The Plenum Series on Demographic Methods and Population Analysis.
  Plenum Press, New York, 1998.
\newblock Translated from the 1939 French edition and with an introduction by
  David P. Smith and H{\'e}l{\`e}ne Rossert.

\bibitem{Mich_Fom}
A.~S. Miscenko and A.~T. Fomenko.
\newblock A generalized {L}iouville method for the integration of {H}amiltonian
  systems.
\newblock {\em Funkcional. Anal. i Prilo zen.}, 12(2):46--56, 96, 1978.

\bibitem{NoumiYamada1998}
M.~Noumi and Y.~Yamada.
\newblock Affine {W}eyl groups, discrete dynamical systems and {P}ainlev{\'e}
  equations.
\newblock {\em Comm. Math. Phys.}, 199(2):281--295, 1998.

\bibitem{KKQTV}
P.~H. van~der Kamp, T.~E. Kouloukas, G.~R.~W. Quispel, D.~T. Tran, and
  P.~Vanhaecke.
\newblock Integrable and superintegrable systems associated with multi-sums of
  products.
\newblock {\em Proc. R. Soc. Lond. Ser. A Math. Phys. Eng. Sci.},
  470(2172):20140481, 23, 2014.

\bibitem{Veselov_maps}
A.~P. Veselov.
\newblock Integrable maps.
\newblock {\em Russian Mathematical Surveys}, 46(5):1--51, oct 1991.

\bibitem{VS}
A.~P. Veselov and A.~B. Shabat.
\newblock A dressing chain and the spectral theory of the {S}chr\"odinger
  operator.
\newblock {\em Funktsional. Anal. i Prilozhen.}, 27(2):1--21, 96, 1993.

\bibitem{Volterra}
V.~Volterra.
\newblock {\em Le\c cons sur la th\'eorie math\'ematique de la lutte pour la
  vie}.
\newblock Les Grands Classiques Gauthier-Villars. [Gauthier-Villars Great
  Classics]. \'Editions Jacques Gabay, Sceaux, 1990.
\newblock Reprint of the 1931 original.

\end{thebibliography}
\end{document}